%% file: main.tex
\newif\ifacm
\acmfalse

\ifacm
\documentclass[manuscript,10pt]{acmart}
\AtBeginDocument{%
  }

\setcopyright{none}
\renewcommand\footnotetextcopyrightpermission[1]{}
\acmYear{2023}

\else

\documentclass[a4paper, 11pt]{article}
\usepackage[utf8]{inputenc}
\usepackage[margin=1in]{geometry}
\usepackage[T1]{fontenc}
\usepackage{amssymb}
\usepackage{authblk}
\fi

\usepackage{graphicx}
\usepackage{algorithm,algorithmic}
\usepackage{hyperref}
\usepackage{framed}
\usepackage{amsfonts}
\usepackage{amsmath,amsthm}
\usepackage{mathtools}
\usepackage{capt-of}
\usepackage{xcolor}
\usepackage{float}
\usepackage{enumitem}
\setlist{nosep}
\usepackage[many]{tcolorbox}
\usetikzlibrary{shadows}
\usepackage{soul}

\newcommand{\ignore}[1]{}

\newcommand{\msf}{\mathsf}
\newcommand{\mc}{\mathcal}
\newcommand{\mbf}{\mathbf}
\newcommand{\sid}{\mathsf{sid}}
\newcommand{\Cl}{\mathsf{Cl}}

\newcommand\myurcorner{\hspace{-0.3em}\raise-0.2ex\hbox{$\urcorner$}}
\newcommand\myulcorner{\raise-0.2ex\hbox{$\ulcorner$}\hspace{-0.3em}}
\newcommand\extitem{\begin{tiny}$\blacksquare$ \end{tiny}}
\newcommand{\NL}[1]{\textcolor{orange}{#1}}
\newcommand{\TZ}[1]{\textcolor{magenta}{\textbf{Thomas: }{#1}}}

\newtheorem{theorem}{Theorem}
\newtheorem{lemma}{Lemma}
\newtheorem{corollary}{Corollary}
\newtheorem{fact}{Fact}

\newtheorem{definition}{Definition}
\theoremstyle{definition}

\title{Universally Composable Simultaneous Broadcast against a Dishonest Majority and Applications}
\author[1]{Myrto Arapinis}
\author[1]{\'Abel Kocsis}
\author[1]{Nikolaos Lamprou}
\author[2]{Liam Medley}
\author[1]{Thomas Zacharias}

\affil[1]{The University of Edinburgh, UK}
\affil[2]{Royal Holloway University of London, UK}
 \date{}

\begin{document}
\maketitle

\input{introduction}

\input{background}
\input{FBC}

\input{real_TLE}
\input{SBC}
\input{applications}

\paragraph{Acknowledgements.}  Zacharias was supported by Input Output (\url{https://iohk.io}) through their funding of the Edinburgh Blockchain Technology Lab.
\ifacm
\else
\bibliographystyle{alpha}
\fi
\bibliography{references}

\appendix
\input{appendix}
\end{document}

%% file: introduction.tex
\section{Introduction}\label{sec:intro}

Communication over a broadcast channel guarantees consistency of message delivery, in the sense that all honest parties output the same message, even when the sender is malicious. Since its introduction by Pease et al.~\cite{PeaseSL80}, broadcast has been a pivotal concept in fault tolerant distributed computing and cryptography. From a property-based security perspective, broadcast communication dictates that every honest party will output some value (termination) that is the same across all honest parties (agreement) and matches the sender's value, when the sender is honest (validity). The first efficient construction was proposed by Dolev and Strong~\cite{DS82}. In particular, the Dolev-Strong broadcast protocol deploys public-key infrastructure (PKI) to achieve property-based security against an adversary corrupting up to $t<n$ parties, where $n$ is the number of parties. In the context of simulation-based security though, Hirt and Zikas~\cite{HirtZ10} proved that broadcast under $t>\frac{n}{2}$ corruptions (dishonest majority) is impossible, even assuming a PKI. In the model of~\cite{HirtZ10}, the adversary may adaptively corrupt parties within the duration of a round (\emph{non-atomic communication model}). Subsequently, Garay et al.~\cite{GarayKKZ11} showed that in the weaker setting where a party cannot be corrupted in the middle of a round (\emph{atomic communication model}), PKI is sufficient for realizing adaptively secure broadcast against an adversary corrupting up to $t<n$ parties.

An important class of protocols that has attracted considerable attention is the one where broadcast is \emph{simultaneous}, i.e., all senders transmit their messages independently of the messages broadcast by other parties. The concept of simultaneous broadcast (SBC) was put forth by Chor et al.~\cite{ChorGMA85} and has proved remarkably useful in the design of various distributed computing constructions (e.g., multiparty computation, coin flipping, electronic voting, fair bidding). The works of Chor and Rabin~\cite{ChorR87} and Gennaro~\cite{Gennaro00} improve the round complexity of~\cite{ChorGMA85} from linear to logarithmic, and from logarithmic to constant (in $n$), respectively. From a security modeling aspect, Hevia and Micciancio~\cite{HeviaM05} point out the hierarchy between the SBC definitions in~\cite{ChorGMA85,ChorR87,Gennaro00} as \cite{ChorGMA85}$\Rightarrow$\cite{ChorR87}$\Rightarrow$\cite{Gennaro00} (from strongest to weakest). Specifically, the simulation-based definition of~\cite{ChorGMA85} implies sequentially composable security. Under the definition of~\cite{Gennaro00}, Faust et al.~\cite{FaustKL08} present a construction with a performance gain in the presence of repeated protocol runs. All the aforementioned SBC solutions~\cite{ChorGMA85,ChorR87,Gennaro00,FaustKL08} achieve security that tolerates $t<\frac{n}{2}$ corruptions (honest majority).

The concept of SBC that retains security under concurrent executions has been formally investigated by Hevia~\cite{Hevia06}. Namely,~\cite{Hevia06} proposes a formal SBC treatment in the state-of-the-art Universal Composability (UC) framework of Canetti~\cite{UC} along with a construction that has constant round complexity and is secure assuming an honest majority. Composable security is crucial for any broadcast channel functionality that serves as a building block for distributed protocol design and is a primary goal of our work.

\begin{figure}
    \centering
    \includegraphics[scale=0.38]{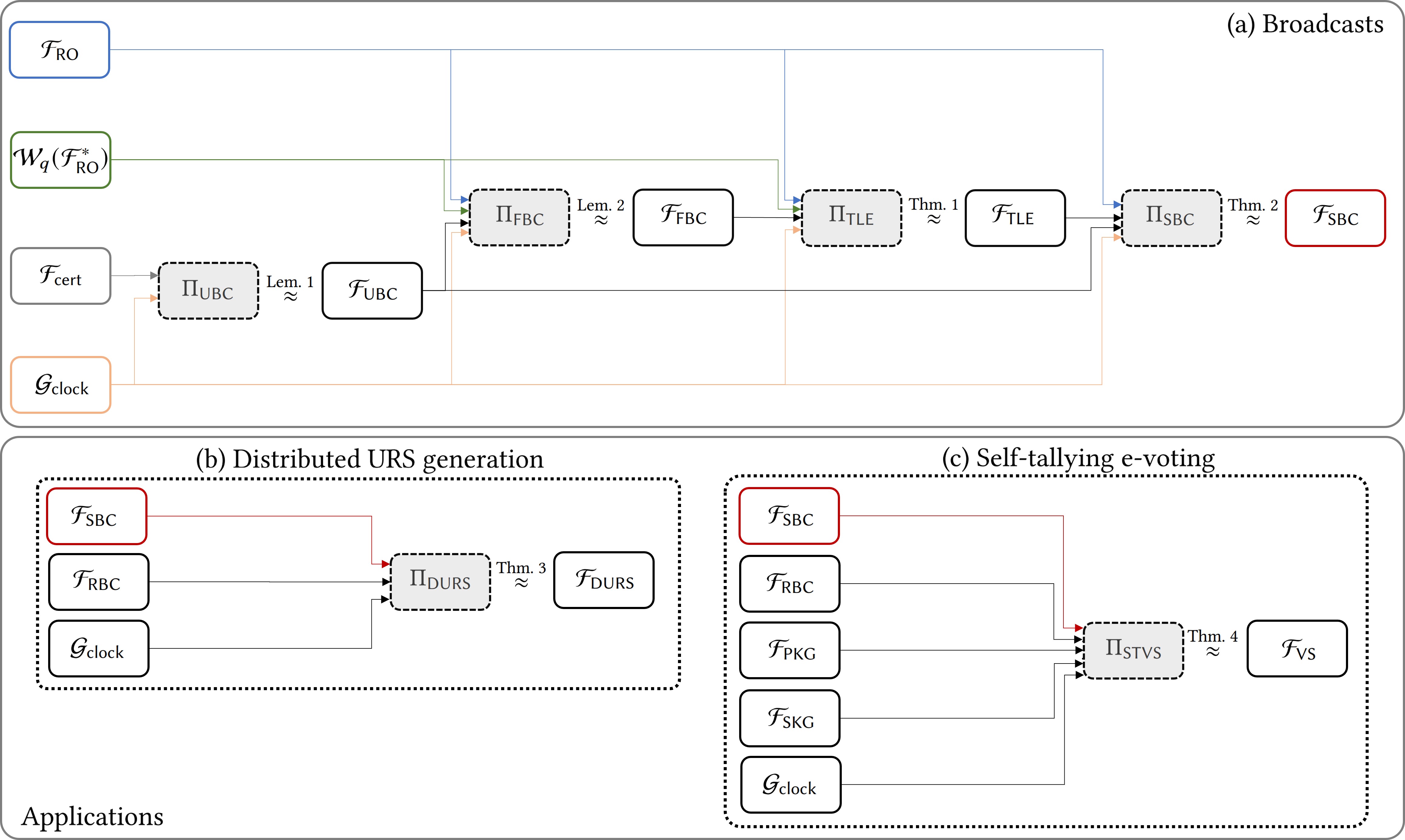}
    \caption{Overview of the paper's contributions. We denote $\Pi_\msf{X}$ the $\msf{X}$, and $\mc{F}_\msf{X}$ the ideal functionality capturing the security requirements for $\msf{X}$, and in UC fashion we write $\Pi_\msf{X} \approx \mc{F}_\msf{X}$ to denote that the protocol $\Pi_\msf{X}$ realizes the ideal functionality $\mc{F}_\msf{X}$ (thus, ensuring the same security properties). Our results rely on the following hybrid functionalities: (i) the global clock $\mc{G}_\msf{clock}$, (ii) the random oracles $\mc{F}_\msf{RO}$ and $\mc{F}^*_\msf{RO}$, (iii) the wrapper $\mc{W}_q(\cdot)$, (iv) the certification $\mc{F}_\msf{cert}$ modelling a PKI, (v) the relaxed broadcast $\mc{F}_\msf{RBC}$ that allows a single message to be broadcast in an unfair manner (and can be realized via $\mc{F}_\msf{cert}$ and $\mc{G}_\msf{clock}$, cf. Fact~\ref{fact:RBC}), (vi) the public key threshold key generation $\mc{F}_\msf{PKG}$, and (vii) the signature key generation $\mc{F}_\msf{SKG}$.} 
    \label{fig:overview}
\end{figure}

\indent\emph{Our contributions.} We explore the SBC problem in the context of UC security against a dishonest majority. We improve the results of~\cite{Hevia06} both from a definitional and a security aspect. In more detail, we achieve the following improvements (summarised in Figure~\ref{fig:overview}(a)):
\begin{itemize}[leftmargin=*]
    \item We define a new SBC functionality that abstracts communication given an agreed broadcast period, outside of which all broadcast operations are discarded. Our functionality captures (i) \emph{simultaneity}: corrupted senders broadcast without having any information about honest senders' messages; (ii) \emph{liveness}: after the broadcast period ends, termination is guaranteed (with some delay) without the requirement of full participation by all parties. We stress that the latter property is not captured by the functionality of~\cite{Hevia06}, as the adversary (simulator) may wait indefinitely until it allows termination of the execution which happens only after all (honest and corrupted) senders have transmitted their value.
    \item We provide a construction that realizes our SBC functionality in an optimal way, that is, it preserves \emph{UC security against a Byzantine adversary that can adaptively corrupt up to $t<n$ parties in the non-atomic communication model}. To overcome the impossibility result of~\cite{HirtZ10}, besides PKI, we deploy (i) adaptively secure time-lock encryption (TLE) in the UC setting; (ii) a programmable random oracle (RO). Specifically, via TLE (and the programmable RO), senders perform (equivocable) encryptions of their message that can be decrypted by any party when the decryption time comes, with some delay upon the end of the broadcast period. It is easy to see that the semantic security of the TLE ciphertexts that lasts throughout the broadcast period guarantees simultaneity. The broadcast period is set dynamically, by having the first sender of the session (as scheduled by the environment) ``wake up'' the other parties via the broadcast of a special message.  
   
    Although utilizing a programmable RO is a standard technique to enable equivocation in simulation-based security (e.g., in~\cite{Nielsen02,BaumDDNO21,ALZ21,CGZ21}), TLE with adaptive UC security is not currently available in the literature. To construct it, we rely on the findings of the following recent papers:
    \begin{enumerate}
        \item The work of Arapinis et al.~\cite{ALZ21} that provides a UC treatment of TLE and a protocol that is secure against a static adversary.
        \item The work of Cohen et al.~\cite{CGZ21} that studies the concept of broadcast and fairness in the context of resource-restricted cryptography~\cite{GarayKOPZ20}. In more detail, Cohen et al. prove that time-lock puzzles (a notion closely related to TLE) and a programmable RO are sufficient for building simulation-based secure broadcast (with limited composability) against an adaptive adversary that corrupts up to $t<n$ parties in the non-atomic model. Interestingly, they also show that neither time-lock puzzles nor programmable ROs alone are enough to achieve such level of broadcast security. In~\cite{CGZ21}, standard broadcast encompasses \emph{fairness}, i.e., an adversary that adaptively corrupts a sender after learning her value cannot change this original value. Besides, the weaker notion of \emph{unfair} broadcast has been introduced in~\cite{HirtZ10} and can be realized by the Dolev-Strong protocol~\cite{DS82} (assuming PKI) against $t<n$ adaptive corruptions.
    \end{enumerate}
    Compared to~\cite{ALZ21} and~\cite{CGZ21}, we take the following steps: first, we adapt the fair broadcast (FBC) and unfair broadcast (UBC) functionalities in~\cite{CGZ21} to the UC setting, where multiple senders may perform many broadcasts per round. Then, similar to~\cite{BadertscherMTZ17,GarayKOPZ20,ALZ21},  we model resource-restriction in UC via wrapper that allows all parties to perform up to a number of RO queries per round. Next, we revisit the FBC protocol in~\cite{CGZ21} by using the TLE algorithms of~\cite{ALZ21} instead of an arbitrary time-lock puzzle and show that our instantiation UC-realizes our FBC functionality. Finally, we prove that by deploying the TLE protocol of~\cite{ALZ21} over our FBC functionality is sufficient to provide an adaptively secure realization of the TLE functionality in~\cite{ALZ21}. We view the construction of the first adaptively UC secure TLE protocol as a contribution of independent interest. We refer the reader to Section~\ref{subsec:UFBC_FBC} for a detailed discussion of the key subtleties to the design of our composably secure (un)fair broadcast protocols.
  \item The SBC construction in~\cite{Hevia06} is over the synchronous communication functionality in~\cite{UC}. As~\cite{KatzMTZ13} shows, this functionality does not provide the guarantees expected of a synchronous network (specifically, termination). These limitations are lifted when relying on a (global) clock functionality~\cite{KatzMTZ13}, as we do in our formal treatment. The use of a global clock is the standard way to model loose synchronicity in UC: every clock tick marks the advance of the execution rounds while within a round, communication is adversarially scheduled by the environment.
\end{itemize}

Armed with our construction, we present two interesting applications of SBC that enjoy adaptive UC security. Namely,
\begin{itemize}[leftmargin=*]
    \item \emph{Distributed uniform random string generation (Figure~\ref{fig:overview}(b)).} We devise a protocol where a set of parties contribute their share of randomness via our SBC channel. After some delay (upon the end of the broadcast period), the honest parties agree on the XOR of the shares they received as a common uniform random string. We call this \emph{delayed uniform random string} (DURS) generation.
    \item \emph{Self-tallying e-voting (Figure~\ref{fig:overview}(c)).} Self-tallying voting systems (STVSs) constitute a special class of decentralized electronic voting systems put forth by Kiayias and Yung~\cite{KiayiasY02}, where the voters can perform the tally themselves without the need for a trusted tallying authority. Most existing efficient STVSs~\cite{KiayiasY02,Groth04,SzepieniecP15} require a trusted party to ensure election fairness (i.e., no partial results are leaked before the end of the vote casting period). We remove the need of a trusted party in self-tallying elections by modifying the construction in~\cite{SzepieniecP15} (shown secure in the UC framework). In particular, we deploy our SBC channel for vote casting instead of a bulletin board used in the original protocol.
\end{itemize}

%% file: background.tex
\section{Background}\label{sec:background}
\subsection{Network model}\label{subsec:background_network}
We consider synchronous point-to-point communication among $n$ parties from a party set $\mbf{P}$, where protocol execution is carried out in rounds. The adversary is Byzantine and may adaptively corrupt any number of $t<n$ parties. The corruption is w.r.t. the strong \emph{non-atomic communication model} (cf.~\cite{HirtZ10,CGZ21}) where the adversary may corrupt parties in the middle of a round, e.g., it may corrupt a sender after receiving the sender's message.
\subsection{The UC framework}\label{subsec:background_UC}
Universal Composability (UC), introduced by Canetti in~\cite{UC}, is a state-of-the-art framework for the formal study of protocols that should remain secure under concurrent executions. In UC, security is captured via the \emph{real world/ideal world} paradigm as follows.
\begin{itemize}[leftmargin=*]
    \item In the ideal world, an \emph{environment} $\mc{Z}$ schedules the execution and provides inputs to the parties that are \emph{dummy}, i.e., they simply forward their inputs to an \emph{ideal functionality} $\mc{F}$, which abstracts the studied security notion (e.g., secure broadcast). The functionality is responsible for carrying out the execution given the forwarded input and returns to the party some output along with a destination identity $ID$, so that the dummy party forwards the output to $ID$. By default, we assume that the destination is $\mc{Z}$, unless specified explicitly. The execution is carried out in the presence of an ideal adversary $\mc{S}$, the \emph{simulator}, that interacts with $\mc{F}$ and $\mc{Z}$ and controls corrupted parties. We denote by $\textsc{EXEC}_{\mc{F},\mc{S},\mc{Z}}$ the output of $\mc{Z}$ after ending the ideal world execution.
    \item In the real world, $\mc{Z}$ schedules the execution and provides inputs as previously, but now the parties actively engage in a joint computation w.r.t. the guidelines of some protocol $\Pi$ (e.g., a broadcast protocol). The execution is now in the presence of a real (Byzantine) adversary $\mc{A}$ that interacts with $\mc{Z}$ and may (adaptively) corrupt a number of parties. We denote by $\textsc{EXEC}_{\Pi,\mc{A},\mc{Z}}$ the output of $\mc{Z}$ after ending the real world execution.
\end{itemize}
\begin{definition}
    We say that a protocol $\Pi$ \emph{UC-realizes} the ideal functionality $\mc{F}$ if for every real world adversary $\mc{A}$ there is a simulator $\mc{S}$ such that for every environment $\mc{Z}$, the distributions $\textsc{EXEC}_{\mc{F},\mc{S},\mc{Z}}$ and $\textsc{EXEC}_{\Pi,\mc{A},\mc{Z}}$ are computationally indistinguishable.
\end{definition}
According to the UC Theorem, the UC security of $\Pi$ implies that $\Pi$ can be replaced by $\mc{F}$ in any protocol that invokes $\Pi$ as a subroutine. Besides, a protocol may use as subroutine a functionality that abstracts some setup notion (e.g., PKI, a random oracle, or a global clock). These setup functionalities maybe \emph{global}, in the sense that share their state across executions of multiple protocols~\cite{globalUC}. If a protocol utilizes a set of functionalities $\{\mc{F}_1,\ldots,\mc{F}_k\}$, then we say that its UC security is argued in the \emph{$(\mc{F}_1,\ldots,\mc{F}_k)$-hybrid model}.

\subsection{Hybrid functionalities}\label{subsec:background_hybrid}
Throughout the paper, $\lambda$ denotes the security parameter and $\msf{negl}(\cdot)$ a negligible function.\medskip

\noindent\textbf{The global clock functionality.} The global clock (cf.~\cite{KatzMTZ13,BadertscherMTZ17} and Figure~\ref{fig:clock}) can be read at any moment by any party. For each session, the clock advances only when all the involved honest parties and functionalities in the session make an \textsc{Advance\_Clock} request.

\begin{figure}[H]
\begin{tcolorbox}[enhanced, colback=white, arc=10pt, drop shadow southeast]
\noindent\emph{\underline{The global clock functionality \(\mathcal{G}_\msf{clock}(\mathbf{P}, \mathbf{F})\).}}\\[5pt]
\begin{small}
The functionality manages the set $\mbf{P}$ of registered identities, i.e., parties $P=(\msf{pid},\sid)$ and the set $\mbf{F}$ of registered functionalities (with their session identifier) $(\mc{F},\sid)$. For every $\sid$, let $\mbf{P}_\sid=\{(\cdot,\sid)\in\mbf{P}\}\cap\{P\in\mbf{P}\mid P \mbox{ is honest}\}$ and $\mbf{F}_\sid=\{(\cdot,\sid)\in\mbf{F}\}$.

For each session \(\sid\), the functionality initializes the clock variable \(\textsf{Cl}_\msf{sid} \leftarrow 0\) and the set of advanced entities per round as \(L_\textsf{sid.adv} \leftarrow \emptyset\).\\[2pt]
\extitem Upon receiving \((\textsc{sid}_C, \textsc{Advance\_Clock})\) from \(P\in\mathbf{P}_\sid\), if \(P\not\in L_\textsf{sid.adv}\), then it adds \(P\) to \(L_\textsf{sid.adv}\). If \(L_\textsf{sid.adv} = \mathbf{P}_\sid\cup \mathbf{F}_\sid\), then it updates \(\textsf{Cl}_\msf{sid}\leftarrow \textsf{Cl}_\msf{sid} + 1\), resets \(L_\textsf{sid.adv} \leftarrow \emptyset\). It forwards \((\textsc{sid}_C, \textsc{Advanced\_Clock}, P)\) to $\mc{A}$.\\[2pt]
\extitem Upon receiving \((\textsc{sid}_C, \textsc{Advance\_Clock})\) from $\mathcal{F}$ in a session $\sid$ such that $(\mc{F},\sid)\in\mathbf{F}$, if \((\mc{F},\sid)\not\in L_\textsf{adv}\), then it adds \((\mc{F},\sid)\) to \(L_\textsf{sid.adv}\). If \(L_\textsf{sid.adv} = \mathbf{P}_\sid\cup \mathbf{F}_\sid\), then it updates \(\textsf{Cl}_\msf{sid}\leftarrow \textsf{Cl}_\msf{sid} + 1\), resets \(L_\textsf{sid.adv} \leftarrow \emptyset\). It sends \((\textsc{sid}_C, \textsc{Advanced\_Clock}, \mc{F})\) to this instance of $\mc{F}$.\\[2pt]
\extitem Upon receiving \((\textsc{sid}_C, \textsc{Read\_Clock})\) from any participant (including the environment on behalf of a party, the adversary, or any ideal (shared or local) functionality), it sends \((\textsc{sid}_C, \textsc{Read\_Clock}, \textsf{Cl}_\sid)\) to this participant, where $\sid$ is the sid of the calling instance.
\end{small}
\end{tcolorbox}

\captionof{figure}{The global clock functionality \(\mathcal{G}_\msf{clock}(\mathbf{P}, \mathbf{F})\) interacting with the parties of the set \(\mathbf{P}\), the functionalities of the set \(\mathbf{F}\), the environment \(\mathcal{Z}\) and the adversary \(\mathcal{A}\).}
\label{fig:clock}

\end{figure}

\noindent\textbf{The random oracle functionality.} The RO functionality (cf.~\cite{Nielsen02} and Figure~\ref{fig:RO}) can be seen as a trusted source of random input. Given a query, it returns a random value. It also updates a local variable \(L_\mathcal{H}\) in order to return the same value to similar queries. This functionality can be seen as the "idealization" of a hash function.

\begin{figure}[H]
\begin{tcolorbox}[enhanced, colback=white, arc=10pt, drop shadow southeast]
\noindent\emph{\underline{The random oracle functionality \(\mathcal{F}_\msf{RO}( A, B)\).}}\\[5pt]
\begin{small}
The functionality initializes a list \(L_\mathcal{H} \leftarrow \emptyset\).\\[2pt]
\extitem Upon receiving \((\sid, \textsc{Query}, x)\) from any party $P$, if \(x\in A\), then:
\begin{enumerate}
\item If there exists a pair \((x, h) \in L_\mathcal{H}\), it returns \((\sid, \textsc{Random\_Oracle}, x, h)\) to \(P\).
\item Else, it picks \(h \in B\) uniformly at random, and it inserts the pair \((x,h)\) to the list \(L_\mathcal{H}\). Then, it returns \((\sid, \textsc{Random\_Oracle}, x, h)\) to \(P\).
\end{enumerate}
\end{small}
\end{tcolorbox}
\captionof{figure}{The random oracle functionality \(\mathcal{F}_\msf{RO}\) with respect to a domain \(A\) and a range \(B\).}
\label{fig:RO}
\end{figure}

\noindent\textbf{The certification functionality.} 
The certification functionality (cf.~\cite{Canetti04} and Figure~\ref{fig:cert}) abstracts a certification scheme which provides signatures bound to \emph{identities}.
It provides commands for signature generation and verification, and is tied to a single party (so each party requires a separate instance). It can be realized via an EUF-CMA secure signature scheme combined with a party acting as a trusted certification authority.\\

\noindent\textbf{The wrapper functionality.} We recall the wrapper functionality from~\cite{ALZ21} in Figure~\ref{fig:wrapper} (here, in the adaptive corruption model), for the special case where the wrapped evaluation functionality is the random oracle $\mc{F}_\msf{RO}$. The wrapper $\mc{W}_{q}$ allows the parties to access $\mc{F}_\msf{RO}$ only up to $q$ times per round (clock tick).

\begin{figure}[H]
\begin{tcolorbox}[enhanced, colback=white, arc=10pt, drop shadow southeast]
\noindent\underline{\emph{The certification functionality $\mc{F}_\msf{cert}^S(\mbf{P})$.}}\\[5pt]
\begin{small}
\extitem Upon receiving $(\sid,\textsc{Sign},M)$ from $S$, it sends $(\sid,\textsc{Sign},M)$ to $\mc{S}$. Upon receiving $(\sid,\textsc{Signature},M,\sigma)$ from $\mc{S}$, it checks that no triple $(M,\sigma,0)$ is recorded. If so, it sends $(\sid,\textsc{Signature},M,\bot)$ to $S$ and halts. Otherwise, it sends $(\sid,\textsc{Signature},M,\sigma)$ to $S$ and adds the triple $(M,\sigma,1)$ to an, initially empty, list $L_\msf{sign}$.\\[2pt]
\extitem Upon receiving $(\sid,\textsc{Verify},M,\sigma)$ from a party $P\in\mbf{P}$, it sends $(\sid,\textsc{Verify},M,\sigma)$ to $\mc{S}$. Upon receiving $(\sid,\textsc{Verified},M,\phi)$ from $\mc{S}$, it does:
\begin{enumerate}
    \item If $(M,\sigma,1)\in L_\msf{sign}$, then it sets $f=1$.
    \item Else, if $S$ is not corrupted and no entry $(M,\sigma',1)$ for any $\sigma'$ is recorded, then it sets $f=0$ and adds $(M,\sigma,0)$ to $L_\msf{sign}$.
    \item Else, if there is an entry $(M,\sigma,f')\in L_\msf{sign}$, then it sets $f=f'$.
    \item Else, it sets $f=\phi$ and adds $(M,\sigma',\phi)$ to $L_\msf{sign}$.%
    \item It sends $(\sid,\textsc{Verified},M,f)$ to $P$.
\end{enumerate}
\end{small}
\end{tcolorbox}
 
\captionof{figure}{The certification functionality $\mc{F}_\msf{cert}$ interacting with a signer $S$, a set of parties $\mbf{P}$ and the simulator $\mc{S}$.} 
\label{fig:cert} 
\end{figure}

\begin{figure}[H]
\begin{tcolorbox}[enhanced, colback=white, arc=10pt, drop shadow southeast]

\noindent\emph{\underline{The wrapper functionality $\mc{W}_{q}(\mc{F}_{\msf{RO}},\mc{G}_{\msf{clock}},\mbf{P})$.}}\\[5pt]
\begin{small}
The functionality maintains the set of corrupted parties, $\mathbf{P}_\msf{corr}$, initialized as empty.\\[2pt]
\extitem Upon receiving $(\sid,\textsc{Evaluate},(x_{1},\ldots,x_{j}))$ from $P\in \mbf{P}\setminus \mathbf{P}_{\msf{corr}}$ it reads the time $\msf{Cl}$ from $\mc{G}_{\msf{clock}}$ and does:
\begin{enumerate}
\item\label{oracleq} If there is not a list $L^{P}$ it creates one, initially as empty. Then it does:
\begin{enumerate}
\item\label{oracleq1a} For every $k$ in $\{1,\ldots,j\}$, it forwards the message $(\sid,\textsc{Evaluate},x_{k})$ to $\mc{F}_{\msf{RO}}$.
\item\label{oracleq1b} When it receives back all the corresponding oracle responses $y_1,\ldots,y_j$, it inserts the tuple-$(\msf{Cl},1)\in L^{P}$.
\item\label{oracleq1c} It sends $(\sid,\textsc{Evaluate},((x_{1},y_{1}),\ldots,(x_{j},y_{j})))$ to $P$.
\end{enumerate} 
\item\label{oracleq2} Else if there is a tuple-$(\msf{Cl},j_{\msf{c}})\in L^{P}$ with $j_{\msf{c}}<q$, then it changes the tuple to $(\msf{Cl},j_{\msf{c}}+1)$, and repeats the above steps~\ref{oracleq1a},\ref{oracleq1c}.
\item\label{oracleq3} Else if there is a tuple-$(\msf{Cl}^{*},j_{\msf{c}})\in L^{P}$ such that $\msf{Cl}^{*}<\msf{Cl}$, it updates the tuple as $(\msf{Cl},1)$, and repeats the above steps~\ref{oracleq1a},\ref{oracleq1b},\ref{oracleq1c}.
\end{enumerate}
\extitem Upon receiving $(\sid,\textsc{Evaluate},(x_{1},\ldots,x_{j}),P)$ from $\mc{S}$ on behalf of $P\in \mbf{P}_{\msf{corr}}$ it reads the time $\msf{Cl}$ from $\mc{G}_{\msf{clock}}$ and repeats steps~\ref{oracleq},\ref{oracleq3} except that it maintains the same list, named $L_{\msf{corr}}$, for all the corrupted parties. %
\end{small}
\end{tcolorbox}

\captionof{figure}{The functionality wrapper $\mc{W}_{q}$ parameterized by a number of queries $q$, functionality $\mc{F}_{\msf{RO}}$, the parties in $\mbf{P}$, the simulator $\mc{S}$, and  $\mc{G}_{\msf{clock}}$.}
\label{fig:wrapper}
\end{figure}

\noindent\textbf{The relaxed broadcast functionality.} In Figure~\ref{fig:RBC}, we present the relaxed broadcast functionality $\mc{F}_\msf{RBC}$ (for a single message) in \cite{GarayKKZ11} that is the stepping stone for realizing unfair broadcast (cf. Subsection~\ref{subsec:UFBC_UBC}) which, in turn, is in the core of the design of the fair and simultaneous broadcast constructions. The functionality captures agreement, but only a weak notion of validity, i.e., if a sender is \emph{always} honest and broadcasts a message $M$, then every honest party will output the value $M$.
In addition, we modify the original description of $\mc{F}_\msf{RBC}$ by forcing the delivery of the message to all parties, when the sender (i) is initially corrupted, or (ii) remains honest in the execution and completes her part by forwarding an \textsc{Advance\_Clock} message. This was implicit in~\cite{GarayKKZ11}.

\begin{figure}[H]
\begin{tcolorbox}[enhanced, colback=white, arc=10pt, drop shadow southeast]
\noindent\emph{\underline{The relaxed broadcast functionality $\mc{F}_\msf{RBC}(\mbf{P})$.}}\\[5pt]
\begin{small}
The functionality initializes a pair of variables $(\msf{Output},\msf{Sender})$ as $(\bot,\bot)$. It also maintains the set of corrupted parties, $\mathbf{P}_\msf{corr}$, initialized as empty.\\[2pt]
\extitem Upon receiving $(\sid,\textsc{Broadcast},M)$ from $P\in\mbf{P}\setminus\mathbf{P}_\msf{corr}$, if $(\msf{Output},\msf{Sender})=(\bot,\bot)$, it records the output-sender pair $(\msf{Output},\msf{Sender})\leftarrow (M,P)$ and sends $(\sid,\textsc{Broadcast},M,P)$ to $\mc{S}$.\\[2pt]
\extitem Upon receiving $(\sid,\textsc{Broadcast},M,P)$ from $\mc{S}$ on behalf of $P\in\mathbf{P}_\msf{corr}$, if $(\msf{Output},\msf{Sender})=(\bot,\bot)$, it sends $(\sid,\textsc{Broadcast},M,P)$ to all parties and $\mc{S}$, and halts.\\[2pt]
\extitem Upon receiving $(\sid,\textsc{Allow},\tilde{M})$ from $\mc{S}$, if $\msf{Sender}\in\mbf{P}_\msf{corr}$, it sends $(\sid,\textsc{Broadcast},$ $\tilde{M},\msf{Sender})$ to all parties and $\mc{S}$, and halts. Otherwise, it ignores the message.\\[2pt]
\extitem Upon receiving $(\sid_C, \textsc{Advance\_Clock})$ from $P\in\mbf{P}\setminus\mathbf{P}_\msf{corr}$, if $\msf{Sender}=P$, it sends $(\sid,\textsc{Broadcast},\msf{Output},\msf{Sender})$ to all parties and $\mc{S}$, and halts. Otherwise, it returns $(\sid_C, \textsc{Advance\_Clock})$ to $P$ with destination identity $\mc{G}_\msf{clock}$. 
\end{small}
\end{tcolorbox}
\caption{The functionality $\mc{F}_\msf{RBC}$ interacting with the parties in $\mbf{P}$ and the simulator $\mc{S}$.}
\label{fig:RBC}
\end{figure}

As presented in \cite{HirtZ10,GarayKKZ11}, $\mc{F}_\msf{RBC}$ can be realized based on the Dolev-Strong protocol~\cite{DS82} and a UC-secure signature scheme. Formally,
\begin{fact}[\cite{HirtZ10,GarayKKZ11}]\label{fact:RBC}
There exists a protocol $\Pi_\msf{RBC}$ that UC-realizes $\mc{F}_\msf{RBC}$ in the $(\mc{F}_\msf{cert},\mc{G}_\msf{clock})$-hybrid model against an adaptive adversary corrupting $t<n$ parties (in the non-atomic model).
\end{fact}

\subsection{Time-lock encryption}\label{subsec:background_TLE}
To realize our secure SBC we will mobilise a special type of encryption, called \emph{time-lock encryption} (TLE). TLE allows one to encrypt a message $M$ for a set amount of time $\tau$. Decryption requires a witness $w$ whose computation is inherently sequential. \cite{ALZ21} provides a UC treatment of the TLE primitive, and a TLE scheme that is UC secure against static adversaries. We will revisit TLE in the presence of adaptive adversaries in Section~\ref{sec:real_TLE}.\medskip

\noindent\textbf{The time-lock encryption  (TLE) functionality.}
In Figure~\ref{fig:TLE}, we present the TLE functionality from~\cite{ALZ21}. Here, $\msf{leak}(\cdot)$ is a function over time slots that captures the timing advantage of the adversary in intercepting the TLE ciphertexts, and $\msf{delay}$ is an integer that express the delay of ciphertext generation. \medskip

\noindent\textbf{The Astrolabous TLE scheme.} We recap the algorithms of the Astrolabous TLE scheme from~\cite{ALZ21}. The scheme utilizes a symmetric-key encryption scheme $\Sigma_\msf{SKE}=(\msf{SKE.Gen},$ $\msf{SKE.Enc},\msf{SKE.Dec})$ and a cryptographic hash function $H(\cdot)$.
\begin{itemize}
    \item The encryption algorithm $\msf{AST.Enc}$ on input the message $M$ and the time-lock difficulty (the security parameter $\lambda$ is implicit from the message size) $\tau_\msf{dec}$ does:
    \begin{enumerate}
        \item It samples an SKE key $k\leftarrow\msf{SKE.Gen}(1^\lambda)$;
        \item It computes $c_{M,k}\leftarrow\msf{SKE.Enc}(k,M)$;
        \item It picks a randomness $r_0||\cdots||r_{q\tau_\msf{dec}-1}\overset{\$}{\leftarrow}\big(\{0,1\}^{\lambda}\big)^{q\tau_\msf{dec}}$;
        \item It computes $c_{k,\tau_\msf{dec}}\leftarrow\big(r_0,r_1\oplus H(r_0),\ldots,k\oplus H(r_{q\tau_\msf{dec}-1})\big)$;
        \item It outputs the ciphertext $c:=(\tau_\msf{dec},c_{M,k},c_{k,\tau_\msf{dec}})$.
    
    \end{enumerate}
        \item The decryption algorithm $\msf{AST.Dec}$ takes as input a ciphertext $c=(\tau_\msf{dec},c_{M,k},c_{k,\tau_\msf{dec}})$ and a decryption witness $w_{\tau_\msf{dec}}=\big(H(r_0),\ldots,H(r_{q\tau_\msf{dec}-1})\big)$ that is computed via $q\tau_\msf{dec}$ sequential hash queries as follows: 
        
        Parse $c_{k,\tau_\msf{dec}}$ and compute $H(r_0)$. Then, for $j=1,\ldots,q\tau_\msf{dec}-1$, compute $H(r_j)$ by (i) performing an $\oplus$ operation between $H(r_{j-1})$ and the $j$-th element of $c_{k,\tau_\msf{dec}}$ denoted as $c_{k,\tau_\msf{dec}}[j]$ to derive $r_j$, and (ii) hashing $r_j$.
        
        Given $c$ and $w_{\tau_\msf{dec}}$, the algorithm $\msf{AST.Dec}$ does:
        \begin{enumerate}
            \item It extracts the SKE key as $k\leftarrow H(r_{q\tau_\msf{dec}-1})\oplus c_{k,\tau_\msf{dec}}[q\tau_{\msf{dec}-1}]$;
            \item It decrypts $M$ as $M\leftarrow\msf{SKE.Dec}(k,c_{M,k})$.
        \end{enumerate}
\end{itemize}

\begin{figure}[H]
\begin{tcolorbox}[enhanced, colback=white, arc=10pt, drop shadow southeast]
\noindent\emph{\underline{The time-lock encryption functionality $\mc{F}^{\msf{leak,delay}}_{\msf{TLE}}(\mbf{P})$.}}\\[5pt]
\begin{small}
The functionality initializes the list of recorded message/ciphertext $L_{\msf{rec}}$ as empty and defines the tag space $\msf{TAG}$. It also maintains the set of corrupted parties, $\mathbf{P}_\msf{corr}$, initialized as empty.\\[2pt]
\extitem Upon receiving $(\sid,\textsc{Enc},M,\tau)$ from $P\not\in \mathbf{P}_{\msf{corr}}$, it reads the time $\msf{Cl}$ and does:
\begin{enumerate}[leftmargin=*]
\item If $\tau <0$, it returns $(\sid,\textsc{Enc},M,\tau,\bot)$ to $P$.
\item It picks $\msf{tag} \overset{\$}{\leftarrow} \msf{TAG}$ and it inserts the tuple $(M,\textsf{Null},\tau,\msf{tag},\msf{Cl},P)\rightarrow L_{\msf{rec}}$.
\item It sends $(\sid,\textsc{Enc},\tau,\msf{tag},\msf{Cl},0^{|M|},P)$ to $\mc{S}$. Upon receiving the token back from $\mc{S}$ it returns $(\sid,\textsc{Encrypting})$ to $P$.
\end{enumerate}
\extitem Upon receiving $(\sid,\textsc{Update},\{(c_{j},\msf{tag}_{j})\}_{j=1}^{p(\lambda)})$ from $\mc{S}$, for all $c_{j}\neq \textsf{Null}$ it updates each tuple $(M_{j},\textsf{Null},\tau_{j},\msf{tag}_{j},\msf{Cl}_{j},P)$ in $L_{\msf{rec}}$ to $(M_{j},c_{j},\tau_{j},\msf{tag}_{j},\msf{Cl}_{j},P)$.\\[2pt]
\extitem Upon receiving $(\sid,\textsc{Update},\{(c_{j},M_{j}, \tau_{j})\}_{j=1}^{p(\lambda)})$ from $\mc{S}$, for all $j\in[p(\lambda)]$ it stores the  tuple $(M_{j},c_{j},\tau_{j},\textsf{Null},0,\textsf{Null})$ in $L_{\msf{rec}}$.\\[2pt]
\extitem Upon receiving $(\sid,\textsc{Retrieve})$ from $P$, it reads the time $\msf{Cl}$ and does:
\begin{enumerate}
    \item For every tuple $(M,\textsf{Null},\tau,\msf{tag},\msf{Cl}',P)\in L_{\msf{rec}}$ such that $\msf{Cl}-\msf{Cl}'\geq \msf{delay}$, it picks $c\overset{\$}{\leftarrow} \{0, 1\}^{p'(\lambda)}$ and updates the tuple as $(M,c,\tau,\msf{tag},\msf{Cl}',P)$.
    \item It sets $\mathcal{C}:= \{(M,c,\tau)\}_{(M,c,\tau,\cdot,\msf{Cl}',P)\in L_{\msf{rec}}:\msf{Cl}-\msf{Cl}'\geq \msf{delay}}$.
    \item It returns $(\sid,\textsc{Encrypted},\mathcal{C})$ to $P$.
\end{enumerate}
\extitem Upon receiving $(\sid,\textsc{Dec},c,\tau)$ from $P\not\in \mathbf{P}_{\msf{corr}}$, if $c\not= \textsf{Null}$:
\begin{enumerate}[leftmargin=*]
\item If $ \tau <0 $, it returns $(\sid,\textsc{Dec},c,\tau,\bot)$ to $P$. Else, it reads the time $\msf{Cl}$ from $\mc{G}_{\msf{clock}}
$ and: 
\begin{enumerate}
\item If $\msf{Cl}<\tau$, it sends $(\sid,\textsc{Dec},c,\tau,\textsc{More\_Time})$ to $P$.
\item If $\msf{Cl}\geq\tau$, then\\[2pt] 
\;-- If there are two tuples $(M_1,c,\tau_{1},\cdot,\cdot,\cdot),(M_{2},c,\tau_{2},\cdot,\cdot,\cdot)$ in $L_\msf{rec}$ such that $M_{1}\neq M_{2}$ and $c\neq \textsf{Null}$ where $\tau \geq \msf{max}\{\tau_{1},\tau_{2}\}$, it returns to $P$ $(\sid,\textsc{Dec},c,\tau,\bot)$.

\;-- If no tuple $(\cdot,c,\cdot,\cdot,\cdot,\cdot)$ is recorded in $L_\msf{rec}$, it sends $(\sid,\textsc{Dec},c,\tau)$ to $\mc{S}$. Upon receiving $(\sid, \textsc{Dec}, c, \tau, M)$ back from $\mc{S}$ it stores $(M, c, \tau, \textsf{Null}, 0, \textsf{Null})$ in $L_{\msf{rec}}$ and returns $(\sid, \textsc{Dec}, c, \tau, M)$ to $P$. %
 
\;-- If there is a unique tuple $(M,c,\tau_\msf{dec},\cdot,\cdot,\cdot)$ in $L_\msf{rec}$, then if $\tau\geq\tau_\msf{dec}$, it returns $(\sid$, $\textsc{Dec}, c, \tau,M)$  to $P$. Else, if $\msf{Cl}<\tau_\msf{dec}$, it returns $(\sid$, $\textsc{Dec}, c, \tau, \textsc{More\_Time})$ to $P$. Else, if $\msf{Cl}\geq\tau_\msf{dec}>\tau$, it returns $(\sid,\textsc{Dec}, c, \tau, \textsc{Invalid\_Time})$ to $P$. 
\end{enumerate}
\end{enumerate}
\extitem Upon receiving $(\sid,\textsc{Leakage})$ from $\mc{S}$, it reads the time $\msf{Cl}$ from $\mc{G}_{\msf{clock}}
$ and returns $(\sid,\textsc{Leakage},(\{(M,c,\tau)\}_{\forall (M,c,\tau,\cdot,\cdot,\cdot)\in L_{\msf{rec}}:\tau\leq \msf{leak}(\msf{Cl})}\cup\{(M,c,\tau,\msf{tag},\msf{Cl},P)\in L_\msf{rec}\}_{\forall P\in\mbf{P}_{\msf{corr}}}))$ to $\mc{S}$.
\end{small}
\end{tcolorbox}

\captionof{figure}{The functionality $\mc{F}^{\msf{leak,delay}}_{\msf{TLE}}$ parameterized by the security parameter $\lambda$, a leakage function $\msf{leak}$, a delay variable $\msf{delay}$, interacting with the parties in $\mbf{P}$, the simulator $\mc{S}$, and global clock $\mc{G}_{\msf{clock}}$.}
\label{fig:TLE}
\end{figure}

Given Astrolabous, $\mc{F}_\msf{TLE}$ is UC-realized in the static corruption model as stated below.
\begin{fact}[\cite{ALZ21}]\label{fact:TLE}
Let $\mc{F}_\msf{BC}$ be the broadcast functionality defined in~\cite{ALZ21}. There exists a protocol that UC-realizes $\mc{F}^{\msf{leak,delay}}_{\msf{TLE}}$ in the $(\mc{W}_q(\mc{F}^*_\msf{RO}),$ $\mc{F}_\msf{RO},\mc{F}_\msf{BC}, \mc{G}_\msf{clock})$-hybrid model against a static adversary corrupting $t<n$ parties, with leakage function $\msf{leak}(\Cl)=\Cl+1$ and $\msf{delay}=1$, where $\mc{F}_\msf{RO}$ and $\mc{F}^*_\msf{RO}$ are distinct random oracles.
\end{fact}

%% file: FBC.tex
\section{UC (un)fair broadcast against dishonest majorities}\label{sec:UFBC}

The prior work of Cohen \emph{et al.}~\cite{CGZ21} studies broadcast fairness in a simulation-based fashion with limited composability. Here, we revisit the concept of broadcast fairness in the setting of UC security, where protocol sessions may securely run concurrently or as subroutines of larger protocols; and in each session, every party can send of multiple messages. We provide a comprehensive formal treatment of the notions of \emph{unfair broadcast} (UBC) and \emph{fair broadcast} (FBC) that will be the stepping stones for the constructions of the following sections.

\subsection{Unfair broadcast definition and realization}\label{subsec:UFBC_UBC}

\noindent\textbf{The UBC functionality.}  We consider a relaxation of FBC, captured by the notion of unfair broadcast introduced in~\cite{HirtZ10}. We present the UBC functionality in Figure~\ref{fig:UBC}. Informally, in UBC, the adversary (simulator) is allowed to receive the sender's message before broadcasting actually happens, and (unlike in FBC) adaptively corrupt the sender to broadcast a message of its preference.
\begin{figure}[H]
\begin{tcolorbox}[enhanced, colback=white, arc=10pt, drop shadow southeast]
\noindent\emph{\underline{The unfair broadcast functionality $\mc{F}_\msf{UBC}(\mbf{P})$.}}\\[5pt]
\begin{small}
The functionality initializes list $L_\msf{pend}$ of pending messages as empty. It also maintains the set of corrupted parties, $\mathbf{P}_\msf{corr}$, initialized as empty.\\[2pt]
\extitem Upon receiving $(\sid,\textsc{Broadcast},M)$ from $P\in\mbf{P}\setminus\mbf{P}_\msf{corr}$, it picks a unique random $\msf{tag}$ from $\{0,1\}^\lambda$, adds the tuple $(\msf{tag},M,P)$ to $L_\msf{pend}$ and sends $(\sid,\textsc{Broadcast},\msf{tag},M,P)$ to $\mc{S}$.\\[2pt]
\extitem Upon receiving $(\sid,\textsc{Broadcast},M,P)$ from $\mc{S}$ on behalf of $P\in\mbf{P}_\msf{corr}$, it sends $(\sid,\textsc{Broadcast},M)$ to all parties and $\mc{S}$.\\[2pt]
\extitem Upon receiving $(\sid,\textsc{Allow},\msf{tag},\tilde{M})$ from $\mc{S}$, if there is a tuple $(\msf{tag},\cdot,P)\in L_\msf{pend}$ such that $P\in\mbf{P}_\msf{corr}$, it does:
\begin{enumerate}[leftmargin=*]
    \item It sends $(\sid,\textsc{Broadcast},\tilde{M})$ to all parties and $(\sid,\textsc{Broadcast},\tilde{M},P)$ to $\mc{S}$.
    \item It deletes $(\msf{tag},\cdot,P)$ from $L_\msf{pend}$.
\end{enumerate}
\extitem Upon receiving $(\sid_C,\textsc{Advance\_Clock})$ from $P\in\mbf{P}\setminus\mbf{P}_\msf{corr}$ it does:
\begin{enumerate}[leftmargin=*]
    \item It reads the time $\Cl$ from $\mc{G}_\msf{clock}$. If this is the first time that $P$ has sent a $(\sid_C,\textsc{Advance\_Clock})$ message during round $\Cl$, then for every $(\msf{tag},M,P)\in L_\msf{pend}$, it does:
\begin{enumerate}
    \item It sends $(\sid,\textsc{Broadcast},M)$ to all parties and $(\sid,\textsc{Broadcast},M,P)$ to $\mc{S}$.
    \item It deletes $(\msf{tag},M,P)$ from $L_\msf{pend}$.
\end{enumerate}
\item It returns $(\sid_C, \textsc{Advance\_Clock})$ to $P$ with destination identity $\mc{G}_\msf{clock}$.
\end{enumerate}
\end{small}
\end{tcolorbox}
\caption{The functionality $\mc{F}_\msf{UBC}$ interacting with the parties in $\mbf{P}$ and the simulator $\mc{S}$.}
\label{fig:UBC}
\end{figure}

\noindent\textbf{The UBC protocol.} In Figure~\ref{fig:real_UBC}, we present a simple protocol that utilizes multiple instances of $\mc{F}_\msf{RBC}$ (cf. Figure~\ref{fig:RBC}) to realize concurrent unfair broadcast executions. The invocation to the $\mc{F}_\msf{RBC}$ instances replaces the composition of multiple Dolev-Strong runs.

By the description of $\Pi_\msf{UBC}$, the Universal Composition Theorem~\cite{UC}, and Fact~\ref{fact:RBC}, we get the following lemma (cf. proof in Appendix~\ref{app:UBC}).
\begin{lemma}\label{lem:real_UBC}
There exists a protocol that UC-realizes $\mc{F}_\msf{UBC}$ in the $(\mc{F}_\msf{cert},\mc{G}_\msf{clock})$-hybrid model against an adaptive adversary corrupting $t<n$ parties.
\end{lemma}

\noindent\textbf{The FBC functionality.} 
Our FBC functionality $\mc{F}^{\Delta,\alpha}_\msf{FBC}$ has the FBC functionality in~\cite{CGZ21} as a reference point, and extends~\cite{CGZ21} to the setting where any party can send of multiple messages per round. In FBC, the adversary (simulator) can receive the sender's message before its broadcasting actually happens. However, even if it adaptively corrupts the sender, the adversary cannot alter the original message that has been ``locked'' as the intended broadcast value.

The functionality is parameterized by two integers: (i) a \emph{delay} $\Delta$, and (ii) an \emph{advantage} $\alpha$ of the simulator $\mc{S}$ to retrieve the broadcast messages compared to the parties. Specifically, if a message is requested to be broadcast at time $\Cl^*$, then $\mc{F}^{\Delta,\alpha}_\msf{FBC}$ will send it to the parties at time $\Cl^*+\Delta$, whereas $\mc{S}$ can obtain it at time $\Cl^*+\Delta-\alpha$.

\begin{figure}[H]
\begin{tcolorbox}[enhanced, colback=white, arc=10pt, drop shadow southeast]
\noindent\emph{\underline{The unfair broadcast protocol $\Pi_\msf{UBC}(\mc{F}_\msf{RBC},\mbf{P})$.}}\\[5pt]
\begin{small}
Every party $P$ maintains two counters $\msf{total}^P,\msf{count}^P$, initialized to $0$. \\[2pt]
\extitem Upon receiving $(\sid,\textsc{Broadcast},M)$ from $\mc{Z}$, the party $P$ does:
\begin{enumerate}[leftmargin=*]
\item She increases $\msf{count}^P$ and $\msf{total}^P$ by $1$.
\item She sends $(\sid,\textsc{Broadcast},M)$ to $\mc{F}^{P,\msf{total}^P}_\msf{RBC}$.
\end{enumerate}
\vspace{2pt}
\extitem Upon receiving $(\sid,\textsc{Broadcast},M^*,P^*)$ from $\mc{F}^{P^*,\cdot}_\msf{RBC}$, the party $P$ forwards $(\sid,\textsc{Broadcast},M^*)$ to $\mc{Z}$.\\[2pt]
\extitem Upon receiving $(\sid_C,\textsc{Advance\_Clock})$ from $\mc{Z}$, the party $P$ reads the time $\Cl$ from $\mc{G}_\msf{clock}$. If this is the first time that she has received a $(\sid_C,\textsc{Advance\_Clock})$ command during round $\Cl$, she does:
\begin{enumerate}[leftmargin=*]
    \item For every $j=1,\ldots,\msf{count}^P$, she sends $(\sid_C,\textsc{Advance\_Clock})$ to $\mc{F}^{P,\msf{total}^P-(\msf{count}^P-j)}_\msf{RBC}$. Namely, $P$ instructs $\mc{F}^{P,\msf{total}^P-(\msf{count}^P-j)}_\msf{RBC}$ to broadcast her $j$-th message for the current round $\Cl$.
    \item She resets $\msf{count}^P$ to $0$.
    \item She forwards $(\sid_C,\textsc{Advance\_Clock})$ to $\mc{G}_\msf{clock}$.
\end{enumerate}
\end{small}
\end{tcolorbox}
\caption{The protocol $\Pi_\msf{UBC}$ with the parties in $\mbf{P}$.}
\label{fig:real_UBC}
\end{figure}
\subsection{Fair broadcast definition and realization}\label{subsec:UFBC_FBC}
%
\begin{figure}[H]
\begin{tcolorbox}[enhanced, colback=white, arc=10pt, drop shadow southeast]
\noindent\emph{\underline{The fair broadcast functionality $\mc{F}^{\Delta,\alpha}_\msf{FBC}(\mbf{P})$.}}\\[5pt]
\begin{small}
The functionality initializes the list $L_\msf{pend}$ of (unlocked) pending messages as empty, the list $L_\msf{lock}$ of locked messages as empty, and a variable $\msf{Output}$ as $\bot$. It also maintains the set of corrupted parties, $\mathbf{P}_\msf{corr}$, initialized as empty.\\[2pt]
\extitem Upon receiving $(\sid,\textsc{Broadcast},M)$ from $P\in\mbf{P}\setminus\mbf{P}_\msf{corr}$ or $(\sid,\textsc{Broadcast},M,P)$ from $\mc{S}$ on behalf of $P\in\mbf{P}_\msf{corr}$, it reads the time $\Cl$ from $\mc{G}_\msf{clock}$, picks a unique random $\msf{tag}$ from $\{0,1\}^\lambda$, and adds the tuple $(\msf{tag},M,P,\Cl)$ to $L_\msf{pend}$. Then, it sends $(\sid,\textsc{Broadcast},\msf{tag},P)$ to $\mc{S}$.\\[2pt]
\extitem Upon receiving $(\sid,\textsc{Output\_Request},\msf{tag})$ from $\mc{S}$, it reads the time $\Cl$ from $\mc{G}_\msf{clock}$. If there is a tuple $(\msf{tag},M,P,\Cl^*)\in L_\msf{pend}$ such that $\Cl-\Cl^*=\Delta-\alpha$, it adds $(\msf{tag},M,P,\Cl^*)$ to $L_\msf{lock}$, removes it from $L_\msf{pend}$, and sends $(\sid,\textsc{Output\_Request},\msf{tag},M,P,\Cl^*)$ to $\mc{S}$. \\[2pt]
\extitem Upon receiving $(\sid,\textsc{Corruption\_Request})$ from $\mc{S}$, it sends $(\sid,\textsc{Corruption\_Request},\langle(\msf{tag},M,P,\Cl^*)\in L_\msf{pend}: P\in\mbf{P}_\msf{corr}\rangle)$ to $\mc{S}$.\\[2pt]
\extitem Upon receiving $(\sid,\textsc{Allow},\msf{tag},\tilde{M},\tilde{P})$ from $\mc{S}$, it does:
\begin{enumerate}[leftmargin=*]
    \item If there is no tuple  $(\msf{tag},M,\tilde{P},\Cl^*)$ in $L_\msf{pend}$ or $L_\msf{lock}$, it ignores the message.
        \item If $\tilde{P}\in\mbf{P}\setminus\mbf{P}_\msf{corr}$ or $(\msf{tag},M,\tilde{P},\Cl^*)\in L_\msf{lock}$, it ignores the message.
    \item If $\tilde{P}\in\mbf{P}_\msf{corr}$ and $(\msf{tag},M,\tilde{P},\Cl^*)\in L_\msf{pend}$ (i.e., the message is not  locked), it sets $\msf{Output}\leftarrow \tilde{M}$. If there is no tuple $(\msf{tag},\cdot,\cdot,\cdot)$ in $L_\msf{lock}$, it adds $(\msf{tag},\msf{Output},\tilde{P},\Cl^*)$ to $L_\msf{lock}$ and removes $(\msf{tag},M,\tilde{P},\Cl^*)$ from $L_\msf{pend}$. It sends $(\sid,\textsc{Allow\_OK})$ to $\mc{S}$.
\end{enumerate}
\vspace{2pt}
\extitem Upon receiving $(\sid_C,\textsc{Advance\_Clock})$ from $P\in\mbf{P}\setminus\mbf{P}_\msf{corr}$, it does:
\begin{enumerate}[leftmargin=*]
    \item It reads the time $\Cl$ from $\mc{G}_\msf{clock}$.
    \item Let $L\leftarrow L_\msf{pend} @ L_\msf{lock}$ be the concatenation of the two lists. It sorts $L$ lexicographically w.r.t. the second coordinate (i.e. messages) of its tuples. 
    \item For every tuple $(\msf{tag}^*,M^*,P^*,\Cl^*)\in L$, if $\Cl-\Cl^*=\Delta$, it sends $(\sid,\textsc{Broadcast},M^*)$ to $P$.
    \item It returns $(\sid_C, \textsc{Advance\_Clock})$ to $P$ with destination identity $\mc{G}_\msf{clock}$.
\end{enumerate}
\end{small}
\end{tcolorbox}
\caption{The functionality $\mc{F}^{\Delta,\alpha}_\msf{FBC}$ interacting with the parties in $\mbf{P}$ and the simulator $\mc{S}$, parameterized by delay $\Delta$ and simulator advantage $\alpha$.}
\label{fig:FBC}
\end{figure}
The functionality associates each \textsc{Broadcast} request with a unique random tag, marks the request as ``pending'', and informs $\mc{S}$ of the senders' activity by leaking the tag and the sender's identity to $\mc{S}$. After $\Delta-\alpha$ rounds, $\mc{S}$ can perform an \textsc{Output\_Request} and obtain the message that corresponds to a specific tag. However, at this point and unlike in UBC, the message becomes ``locked'' and $\mc{S}$ cannot alter it with a message of its choice, even if the sender gets adaptively corrupted.
Besides, by performing a \textsc{Corruption\_Request}, $\mc{S}$ can obtain the pending messages of all corrupted parties, so that it can update the state of the corresponding simulated party with the actual pending messages. The simulator may change the original message of a broadcast request with a value of its choice only if (i) the associated sender is corrupted and (ii) the original message is not locked.
The message delivery to the parties happens when the parties forward an \textsc{Advance\_Clock} message for the round that is $\Delta$ time after the broadcast request occurred. The functionality is formally presented in Figure~\ref{fig:FBC}.\\

\input{real_FBC}

%% file: real_FBC.tex
\noindent\textbf{The FBC protocol.} 
The (stand-alone) FBC protocol proposed in~\cite{CGZ21} is not UC secure. In Figure~\ref{fig:real_FBC}, we present our protocol that realizes concurrent fair broadcast executions. As in~\cite{CGZ21}, we deploy (a) UBC, (b) time-lock puzzles (instantiated by the TLE algorithms in~\cite{ALZ21}) to achieve broadcast fairness, and (c) a programmable RO to allow equivocation (also applied in~\cite{Nielsen02,BaumDDNO21,ALZ21}).

In order to construct FBC in a setting with recurring and arbitrary scheduled broadcast executions, several technical issues arise. The key challenge here is to ensure that messages are retrieved by all parties in the same round. Our protocol carefully orchestrates TLE encryption, emission, reception, and TLE decryption of messages broadcast in UBC manner w.r.t. the global clock to achieve this. The UC-secure protocol $\Pi_\msf{FBC}$ encompasses the following key features:
\begin{enumerate}[leftmargin=*]
    \item Resource-restriction is formalized via a wrapper $\mc{W}_q(\mc{F}^*_\msf{RO})$ that allows a party or the adversary to perform up to $q$ parallel queries per round (cf.~\cite{BadertscherMTZ17,GarayKOPZ20,ALZ21} for similar formal treatments).
    \item To take advantage of parallelization that the wrapper offers, parties generate puzzles for creating TLE ciphertexts (and solve the puzzles of the ciphertexts they have received) only when they are about to complete their round. I.e., when receiving an \textsc{Advance\_Clock} command by the environment, they broadcast in UBC manner all their messages (TLE encrypted with difficulty set to $2$ rounds and equivocated) for the current round. Observe that  if without such restriction and allow senders broadcast their messages upon instruction by the environment, then this would lead to a "waste of resources"; so, parties would not be able to broadcast more than $q$ messages per round and/or they would not have any queries left to proceed to puzzle solution.
\end{enumerate}
\begin{enumerate}[leftmargin=*]
\setcounter{enumi}{2}
    \item For realization of $\mc{F}_\msf{FBC}$, a message must be retrieved by all parties in the same round. Hence, we require that parties, when acting as recipients, begin decryption (puzzle solving)  \emph{in the round that follows} the one they received the associated TLE ciphertext. Otherwise, the following may happen: let parties $A$, $B$, and $C$ complete round $\Cl$ first, second, and third, respectively. If $B$ broadcasts an encrypted message $M$, then, unlike $C$, $A$ will have exhausted its available resources (RO queries) by the time she receives $M$. As a result, $C$ is able to retrieve $M$ at round $\Cl+1$ (by making the first set of $q$ RO queries in $\Cl$ and the second set in $\Cl+1$) whereas $A$ not earlier than $\Cl+2$ (by making the first set in $\Cl+1$ and the second in $\Cl+2$). 
    \item The reason that we impose time difficulty of two rounds instead of just one is rather technical. Namely, if it was set to one round, then the number of queries required for puzzle solution is $q$. However, a rushing real-world adversary may choose to waste all of its resources to decrypt a TLE ciphertext \emph{in the same round} that the ciphertext was intercepted. In this case, the simulator would not have time for equivocating the randomness hidden in the underlying puzzle and simulation would fail.
\end{enumerate}
\begin{figure}[H]
\begin{tcolorbox}[enhanced, colback=white, arc=10pt, drop shadow southeast]
\noindent\emph{\underline{The fair broadcast protocol $\Pi_\msf{FBC}(\mc{F}_\msf{UBC},\mc{W}_q(\mc{F}^*_\msf{RO}),\mc{F}_\msf{RO},\mbf{P})$.}}\\[5pt]
\begin{small}
The protocol utilizes the TLE algorithms $(\msf{AST.Enc},\msf{AST.Dec})$ described in Section~\ref{subsec:background_TLE}. Every party $P$ maintains (i) a list $L^P_\msf{pend}$ of messages pending to be broadcast, (ii) a list $L^P_\msf{wait}$ of received ciphertexts waiting to be decrypted, and (iii) a list $L^P$ of messages ready to be delivered.
All three lists are initialized as empty. \\[2pt]
\extitem Upon receiving $(\sid,\textsc{Broadcast},M)$ from $\mc{Z}$, the party $P$ adds $M$ to $L^P_\msf{pend}$.\\[2pt]
\extitem Upon receiving $(\sid,\textsc{Broadcast},(c^*,y^*))$ from $\mc{F}_\msf{UBC}$, the party $P$ reads the time $\Cl$ from $\mc{G}_\msf{clock}$ and adds $(c^*,y^*,\Cl)$ to $L^P_\msf{wait}$.\\[2pt]
\extitem Upon receiving $(\sid,\textsc{Advance\_Clock})$ from $\mc{Z}$, the party $P$ reads the time $\Cl$ from $\mc{G}_\msf{clock}$. If this is the first time that $P$ has received $(\sid,\textsc{Advance\_Clock})$ for time $\Cl$, she does:
\begin{enumerate}[leftmargin=*]
    \item For every $M$ in $L^P_\msf{pend}$, she picks puzzle randomness $r^M_0||\cdots||r^M_{2q-1}\overset{\$}{\leftarrow}\big(\{0,1\}^{\lambda}\big)^{2q}$.
    \item For every $(c^*,y^*,\Cl-1)$ in $L^P_\msf{wait}$, she parses $c^*$ as $(2,c^*_2,c^*_3)$ and $c^*_3$ as $\big(r^*_0,z^*_1,\ldots,z^*_{2q})$. For every $(c^{**},y^{**},\Cl-2)$ in $L^P_\msf{wait}$, she parses $c^{**}$ as $(2,c^{**}_2,c^{**}_3)$ and $c^*_3$ as $\big(r^{**}_0,z^{**}_1,\ldots,z^{**}_{2q})$. 
    \item She makes all available $q$ queries $Q_0,\ldots,Q_{q-1}$ to $\mc{W}_q(\mc{F}^*_\msf{RO})$ for time $\Cl$ and receives responses $R_0,\ldots,R_{q-1}$, respectively, where
    \begin{itemize}
        \item $Q_0=\big(\cup_{M\in L^P_\msf{pend}}\{r^M_0,\ldots,r^M_{2q-1}\}\big)\bigcup\big(\cup_{(c^*,y^*,\Cl-1)\in L^P_\msf{wait}}\{r^*_0\}\big)\bigcup$\\$\bigcup\big(\cup_{(c^{**},y^{**},\Cl-2)\in L^P_\msf{wait}}\{z^{**}_q\oplus h^{**}_{q-1}\}\big)$.
        \item $R_0=\big(\cup_{M\in L^P_\msf{pend}}\{h^M_0,\ldots,h^M_{2q-1}\}\big)\bigcup\big(\cup_{(c^*,y^*,\Cl-1)\in L^P_\msf{wait}}\{h^*_0\}\big)\bigcup$\\$\bigcup\big(\cup_{(c^{**},y^{**},\Cl-2)\in L^P_\msf{wait}}\{ h^{**}_q\}\big)$.
        \item For $j\geq1$, $Q_j=\big(\cup_{(c^*,y^*,\Cl-1)\in L^P_\msf{wait}}\{z^*_j\oplus h^*_{j-1}\}\big)\bigcup\big(\cup_{(c^{**},y^{**},\Cl-2)\in L^P_\msf{wait}}\{z^{**}_{j+q}\oplus h^{**}_{j+q-1}\}\big)$.
        \item For $j\geq1$, $R_j=\big(\cup_{(c^*,y^*,\Cl-1)\in L^P_\msf{wait}}\{h^*_j\}\big)\bigcup\big(\cup_{(c^{**},y^{**},\Cl-2)\in L^P_\msf{wait}}\{ h^{**}_{j+q}\}\big)$.\footnote{ Namely, the first query includes all puzzle generation queries required for the TLE of every message that will be broadcast by $P$. The $j$-th query includes (i) all $j$-th step puzzle solving queries for decrypting messages received in round $\Cl-1$ and (ii) all $(q+j)$-step puzzle solving queries for decrypting messages received in round $\Cl-2$. The queries are computed as described in Subsection~\ref{subsec:background_TLE}. As a result, the decryption witness for each TLE ciphertext can be computed in two rounds (upon completing all necessary $2q$ hashes).}
    \end{itemize}
\end{enumerate}

%
\begin{enumerate}[leftmargin=*]
\setcounter{enumi}{3}
    \item For every $M$ in $L^P_\msf{pend}$:
    \begin{enumerate}
    \item She chooses a random value $\rho$ from the TLE message space;
    \item She encrypts as $c\leftarrow\msf{AST.Enc}(\rho,2)$ using RO responses $(h^M_0,\ldots,h^M_{2q-1})$ obtained by querying $\mc{W}_q(\mc{F}^*_\msf{RO})$ on $Q_0$.
    \item She queries $\mc{F}_\msf{RO}$ on $\rho$ and receives a response $\eta$.
    \item She computes $y\leftarrow M\oplus\eta$.
    \item She deletes $M$ from $L^P_\msf{pend}$, and sends $(\sid,\textsc{Broadcast},(c,y))$ to $\mc{F}_\msf{UBC}$.
\end{enumerate}
\item For every $(c^{**},y^{**},\Cl-2)$ in $L^P_\msf{wait}$:
\begin{enumerate}
    \item She sets the decryption witness as $w^{**}_2\leftarrow(h^{**}_0,\ldots,h^{**}_{q-1})$.
    \item She computes $\rho^{**}\leftarrow\msf{AST.Dec}(c^{**},w^{**}_2)$.
    \item She queries $\mc{F}_\msf{RO}$ on $\rho^{**}$ and receives a response $\eta^{**}$.
    \item She computes $M^{**}\leftarrow y^{**}\oplus\eta^{**}$ and adds $M^{**}$ to $L^P$.
    \item She deletes $(c^{**},y^{**},\Cl-2)$ from $L^P_\msf{wait}$. 
\end{enumerate}
\item She sorts $L^P$ lexicographically.
\item For every $M^{**}$ in $L^P$, she returns $(\sid,\textsc{Broadcast},M^{**})$ to $\mc{Z}$.
\item She resets $L^{P}$ as empty.
\item She sends $(\msf{sid}_C,\textsc{Advance\_Clock})$ to $\mc{F}_\msf{UBC}$. Upon receiving $(\msf{sid}_C,\textsc{Advance\_Clock})$ from $\mc{F}_\msf{UBC}$, she forwards  $(\msf{sid}_C,\textsc{Advance\_Clock})$ to $\mc{G}_\msf{clock}$ and completes her round.
\end{enumerate}
\end{small}
\end{tcolorbox}
\caption{The protocol $\Pi_\msf{FBC}$ with the parties in $\mbf{P}$.}
\label{fig:real_FBC}
\end{figure}

\vspace{2pt}
The protocol is formally described in Figure~\ref{fig:real_FBC}. The core idea of the construction is the following: to broadcast a message $M$ in a fair manner, the sender chooses a randomness $\rho$ and creates a TLE ciphertext $c$ of $\rho$. Then, she queries the RO on $\rho$ to receive a response $\eta$, computes $M\oplus\eta$, and broadcasts $(c,M\oplus\eta)$ via $\mc{F}_\msf{UBC}$. When decryption time comes, any recipient can decrypt $c$ as $\rho$, obtain $\eta$ via a RO query on $\rho$, and retrieve $M$ by an XOR operation.

In the following lemma (see proof in Appendix~\ref{app:FBC_proof}), we prove that our FBC protocol UC-realizes $\mc{F}^{\Delta,\alpha}_\msf{FBC}$ for delay $\Delta=2$ and advantage $\alpha=2$. Namely, the parties retrieve the messages after two rounds and the simulator two rounds earlier (i.e., in the same round).
\begin{lemma}\label{lem:real_FBC}
The protocol $\Pi_\msf{FBC}$ in Figure~\ref{fig:real_FBC} UC-realizes $\mc{F}^{2,2}_\msf{FBC}$ in the $(\mc{F}_\msf{UBC},\mc{W}_q(\mc{F}^*_\msf{RO}),\mc{F}_\msf{RO},\mc{G}_\msf{clock})$-hybrid model against an adaptive adversary corrupting $t<n$ parties.
\end{lemma}

%% file: real_TLE.tex
\section{UC time-lock encryption against adaptive adversaries}\label{sec:real_TLE}

In this section, we strengthen the main result of~\cite{ALZ21} (cf. Fact~\ref{fact:TLE}), presenting the first UC realization of $\mc{F}_\msf{TLE}$ against adaptive adversaries. Specifically, we prove that the protocol $\Pi_\msf{TLE}$ presented in Frigure~\ref{fig:real_TLE} (over the Astrolabous TLE scheme, cf. Section~\ref{subsec:background_TLE}) is UC secure when deploying $\mc{F}_\msf{FBC}$ as the hybrid functionality that establishes communication among parties. In more details, the TLE construction in~\cite{ALZ21} requires that an encryptor broadcasts her TLE ciphertext to all other parties to allow them to begin solving the associated time-lock puzzle for decryption.

\begin{tcolorbox}[enhanced, colback=white, arc=10pt, drop shadow southeast]
\noindent\emph{\underline{The TLE protocol $\Pi_\msf{TLE}(\mc{F}^{\Delta,\alpha}_\msf{FBC},\mc{W}_q(\mc{F}^*_\msf{RO}),\mc{F}_\msf{RO},\mbf{P})$.}}\\[5pt]
\begin{small}
The protocol utilizes the TLE algorithms $(\msf{AST.Enc},\msf{AST.Dec})$ described in Section~\ref{subsec:background_TLE} and specifies a tag space $\msf{TAG}$. Every party $P$ maintains (i) a list $L^P_\msf{rec}$ of recorded messages/ciphertexts, and (ii) a list $L^P_\msf{puzzle}$ of the recorded oracle queries/responses for puzzle solving.
Both lists are initialized as empty.\\[2pt]
\extitem Upon receiving $(\sid,\textsc{Enc},M,\tau)$ from $\mc{Z}$, the party $P$ does:
\begin{enumerate}
    \item If $\tau<0$, she returns $(\sid,\textsc{Enc},M,\tau,\bot)$ to $\mc{Z}$.
    \item Else, she reads the time $\Cl$ from $\mc{G}_\msf{clock}$. Then, she picks $\msf{tag}\overset{\$}{\leftarrow}\msf{TAG}$, adds $(M,\msf{Null},\tau,\msf{tag},\Cl,0)$ to $L^P_\msf{rec}$ and returns $(\sid,\textsc{Encrypting})$ to $\mc{Z}$.
\end{enumerate}
\vspace{2pt}
\extitem Upon receiving $(\sid,\textsc{Advance\_Clock})$ from $\mc{Z}$, the party $P$ reads the time $\Cl$ from $\mc{G}_\msf{clock}$. If this is the first time that $P$ has received $(\sid,\textsc{Advance\_Clock})$ for time $\Cl$, she does:
\begin{enumerate}
    \item She sends $(\sid_C,\textsc{Advance\_Clock})$ to $\mc{F}^{\Delta,\alpha}_\msf{FBC}$ and receives its response, which is a sequence $(\sid,\textsc{Broadcast},(c_1,\tau_1)),\ldots,$ $(\sid,\textsc{Broadcast},(c_{j_\Cl},\tau_{j_\Cl}))$ of broadcast messages (delayed by $\Delta$ rounds), along with a $(\sid_C,\textsc{Advance\_Clock})$ message (with destination identity $\mc{G}_\msf{clock}$).
    \item For $k=1,\ldots,j_\Cl$, she does:
    \begin{enumerate}
        \item She parses $c_k$ as $(c^k_1,c^k_2,c^k_3)$, and $c^k_1$ as $(c^k_{1,1},c^k_{1,2},c^k_{1,3})$.
        \item She sets $\tau^k_\msf{dec}\leftarrow c^k_{1,1}$  and parses $c^k_{1,3}$ as $(z^k_0,z^k_1,\ldots,z^k_{q\tau^k_\msf{dec}})$.
        \item She picks $\msf{tag}_k\overset{\$}{\leftarrow}\msf{TAG}$ and adds $\big(\msf{tag}_k,c_k,\tau_k,(z^k_0,z^k_1,\ldots,z^k_{q\tau^k_\msf{dec}-1}),\tau^k_\msf{dec}\big)$ to $L^P_\msf{puzzle}$.
    \end{enumerate}
    \item She runs the $\mathtt{ENCRYPT\&SOLVE}$ procedure described below,  updating $L^P_\msf{rec}$ and $L^P_\msf{puzzle}$ accordingly.
    \item For every $(M,c,\tau,\msf{tag},\Cl,0)\in L^P_\msf{rec}$ such that $c\neq\msf{Null}$, she sends $(\sid,\textsc{Broadcast},(c,\tau))$ to $\mc{F}^{\Delta,\alpha}_\msf{FBC}$ and updates the tuple  $(M,c,\tau,\msf{tag},\Cl,0)$ as  $(M,c,\tau,\msf{tag},\Cl,1)\in L^P_\msf{rec}$.
    \item She forwards $(\sid_C,\textsc{Advance\_Clock})$ to $\mc{G}_\msf{clock}$.
    
\end{enumerate}
\vspace{2pt}
\extitem Upon receiving $(\sid, \textsc{Retrieve})$ from $\mc{Z}$, the party $P$ reads the time $\Cl$ from $\mc{G}_\msf{clock}$ and returns $(\sid, \textsc{Encrypted}, \{(M,c,\tau) \mid \exists\Cl^*:(m,c,\tau,\cdot,\Cl^*,1) \in L^P_\msf{rec} \wedge \Cl-\Cl^*\geq \Delta+1\})$ to $\mc{Z}$.\\[2pt]

\end{small}
\end{tcolorbox}

\begin{figure}[H]
\begin{tcolorbox}[enhanced, colback=white, arc=10pt, drop shadow southeast]
\begin{small}
\extitem Upon receiving $(\sid,\textsc{Dec},c,\tau)$ from $\mc{Z}$, the party $P$ does:
\begin{enumerate}
    \item If $\tau<0$, she returns $(\sid,\textsc{Dec},c,\tau,\bot)$ to $\mc{Z}$.
    \item She reads the time $\Cl$ from $\mc{G}_\msf{clock}$.
    \item If $\Cl<\tau$, she returns $(\sid,\textsc{Dec},c,\tau,\textsc{More\_Time})$ to $\mc{Z}$.
    \item If there is no tuple $(\cdot,c,\cdot,\cdot,0)$ in $L^P_\msf{puzzle}$, she returns $(\sid,\textsc{Dec},c,\tau,\bot)$ to $\mc{Z}$.
    \item If there is a tuple $\big(\msf{tag},c,\tau^*,(\zeta_0,\zeta_1,\ldots,\zeta_{q\tau_\msf{dec}-1}),0\big)$ in $L^P_\msf{puzzle}$, she does:
    \begin{enumerate}
     \item If $\tau<\tau^*\leq\Cl$, she returns $(\sid,\textsc{Dec},c,\tau,\textsc{Invalid\_Time})$ to $\mc{Z}$.
     \item She parses $c$ as $(c_1,c_2,c_3)$.
     \item She sets the decryption witness as $w_{\tau_\msf{dec}}\leftarrow(\zeta_0,\zeta_1,\ldots,\zeta_{q\tau_\msf{dec}-1})$.
     \item She computes $\rho\leftarrow\msf{AST.Dec}(c_1,w_{\tau_\msf{dec}})$.
     \item She queries $\mc{F}_\msf{RO}$ on $\rho$ and receives a response $\eta$.
     \item She computes $M\leftarrow c_1\oplus\eta$.
      \item She queries $\mc{F}_\msf{RO}$ on $\rho||M$ and receives a response $c'_3$.
      \item If $c'_3\neq c_3$, she returns $(\sid,\textsc{Dec},c,\tau,\bot)$ to $\mc{Z}$. Else, she returns $(\sid,\textsc{Dec},c,\tau,M)$ to $\mc{Z}$.
 \item[]    
    \end{enumerate}

\end{enumerate}
\end{small}
\noindent\emph{\underline{The procedure $\mathtt{ENCRYPT\&SOLVE}$}}\\[2pt]
\begin{small}
\begin{enumerate}
    \item For every $(M,\msf{Null},\tau,\msf{tag},\Cl,0)\in L^P_\msf{rec}$, 
    \begin{enumerate}
        \item Set $\tau_\msf{dec}\leftarrow\tau-(\Cl+\Delta+1)$ 
        \item Pick puzzle randomness $r^M_0||\cdots||r^M_{q\tau_\msf{dec} -1}\overset{\$}{\leftarrow}\big(\{0,1\}^{\lambda}\big)^{q\tau_\msf{dec}}$.
    \end{enumerate}

    \item Make all available $q$ queries $Q_0,\ldots,Q_{q-1}$ to $\mc{W}_q(\mc{F}^*_\msf{RO})$ and receive responses $R_0,\ldots,R_{q-1}$, respectively as follows:
        \begin{itemize}
        \item $Q_0=\big(\cup_{(M,\msf{Null},\tau,\cdot,\Cl,0)\in L^P_\msf{rec}}\{r^M_0,\ldots,r^M_{q\tau_\msf{dec} -1}\}\big)\bigcup$\\
        $\bigcup\big(\cup_{(\msf{tag}^*,c^*,\tau^*,(\zeta^*_0,\zeta^*_1,\ldots,\zeta^*_{q\tau^*_\msf{dec}-1}),t^*)\in L^P_\msf{puzzle}:t^*\neq0}\{\zeta^*_{q(\tau^*_\msf{dec}-t^*)}\oplus \zeta^*_{q(\tau^*_\msf{dec}-t^*)-1}\}\big)$. For the special case where $t^*=\tau^*_\msf{dec}$, we set $\zeta^*_{-1}=0$. 
          \item $R_0=\big(\cup_{(M,\msf{Null},\tau,\cdot,\Cl,0)\in L^P_\msf{rec}}\{h^M_0,\ldots,h^M_{q\tau_\msf{dec} -1}\}\big)\bigcup$\\
        $\bigcup\big(\cup_{(\msf{tag}^*,c^*,\tau^*,(\zeta^*_0,\zeta^*_1,\ldots,\zeta^*_{q\tau^*_\msf{dec}-1}),t^*)\in L^P_\msf{puzzle}:t^*\neq0}\{h^*_{q(\tau^*_\msf{dec}-t^*)}\}\big)$.
        \item For $j\geq1$, $Q_j=\cup_{(\msf{tag}^*,c^*,\tau^*,(\zeta^*_0,\zeta^*_1,\ldots,\zeta^*_{q\tau^*_\msf{dec}-1}),t^*)\in L^P_\msf{puzzle}:t^*\neq0}\{\zeta^*_{q(\tau^*_\msf{dec}-t^*)+j}\oplus h^*_{q(\tau^*_\msf{dec}-t^*)+j-1}\}$ where $h^*_{q(\tau^*_\msf{dec}-t^*)+j-1} \in R_{j-1}$.
        \item For $j\geq1$, $R_j=\cup_{(\msf{tag}^*,c^*,\tau^*,(\zeta^*_0,\zeta^*_1,\ldots,\zeta^*_{q\tau^*_\msf{dec}-1}),t^*)\in L^P_\msf{puzzle}:t^*\neq0}\{ h^*_{q(\tau^*_\msf{dec}-t^*)+j}\}$.
    \end{itemize}
    \item For every $(M,\msf{Null},\tau,\msf{tag},\Cl,0)\in L^P_\msf{rec}$, do:
    \begin{enumerate}
        \item Set $\tau_\msf{dec}\leftarrow\tau-(\Cl+\Delta+1)$ 
        \item Choose a random value $\rho$ from the TLE message space.
        \item Encrypt as $c_1\leftarrow\msf{AST.Enc}(\rho,\tau_\msf{dec})$ using RO responses $h^M_0,\ldots,h^M_{q\tau_\msf{dec}-1}$ obtained by querying $\mc{W}_q(\mc{F}^*_\msf{RO})$ on $Q_0$.
        \item Query $\mc{F}_\msf{RO}$ on $\rho$ and receive a response $\eta$.
    \item Compute $c_2\leftarrow M\oplus\eta$.
    \item  Query $\mc{F}_\msf{RO}$ on $\rho||M$ and receive a response $c_3$.
    \item Set $c\leftarrow(c_1,c_2,c_3)$.
    \item Update the tuple $(M,\msf{Null},\tau,\msf{tag},\Cl,0)$ as $(M,c,\tau,\msf{tag},\Cl,0)$ in $L^P_\msf{rec}$.
    \end{enumerate}
    \item For every $(\msf{tag}^*,c^*,\tau^*,(\zeta^*_0,\zeta^*_1,\ldots,\zeta^*_{q\tau^*_\msf{dec}-1}),t^*)\in L^P_\msf{puzzle}$ such that $t^*\neq0$, do:
    \begin{enumerate}
        \item For every $j=0,\ldots,q-1$ update as $\zeta^*_{q(\tau^*-t^*)+j}\leftarrow h^*_{q(\tau^*-t^*)+j}$, where $h^*_{q(\tau^*-t^*)+j}$ is included in $R_j$. When this step is completed, the tuple includes the responses of the first $q(\tau^*_\msf{dec}-t^*+1)$ RO queries to $\mc{W}_q(\mc{F}^*_\msf{RO})$.
        \item Update as $t^*\leftarrow t^*-1$. When $t^*$ becomes $0$, $(\zeta^*_0,\zeta^*_1,\ldots,\zeta^*_{q\tau^*_\msf{dec}-1})$ can be parsed as the decryption witness of $c^*$.
    \end{enumerate}
\end{enumerate}
\end{small}
\end{tcolorbox}
\caption{The protocol $\Pi_\msf{TLE}$ with the parties in $\mbf{P}$.}
\label{fig:real_TLE}
\end{figure}

The following theorem shows that FBC is sufficient to guarantee adaptive security of the above TLE protocol. We refer the reader to Appendice~\ref{app:TLE_proof} for the full security proof.
\begin{theorem}\label{thm:real_TLE}
Let $\Delta,\alpha$ be integers s.t. $\Delta\geq\alpha\geq0$. The protocol $\Pi_\msf{TLE}$ in Figure~\ref{fig:real_TLE} UC-realizes $\mc{F}^{\msf{leak},\msf{delay}}_\msf{TLE}$ in the $(\mc{W}_q(\mc{F}^*_\msf{RO}),\mc{F}_\msf{RO},\mc{F}^{\Delta,\alpha}_\msf{FBC},$ $\mc{G}_\msf{clock})$-hybrid model, where $\msf{leak}(\Cl)=\Cl+\alpha$ and $\msf{delay}=\Delta+1$.
\end{theorem}
\ignore{\begin{proof}
\TZ{Nikola, do your magic here.}
\NL{Given the time I did my best. One concern I have is that is hidden the IND-CPA property . I also did not mentioned in the proof for avoiding the confusion (it is not stated in theorem). As a next fix I will include it in the theorem and update the proof as follows: If A manages to solve the TLE puzzle earlier then I can make an adversary that breaks IND-CPA}
\end{proof}}

\ignore{
By applying Theorem~\ref{thm:real_TLE} and Lemma~\ref{lem:real_FBC}, we directly get the following corollary.
}
\ignore{
\begin{corollary}\label{cor:TLE}
There exists a protocol that UC-realises $\mc{F}^{\msf{leak},\msf{delay}}_\msf{TLE}$ in the $(\mc{W}_q(\mc{F}^*_\msf{RO}),\mc{F}_\msf{RO},\mc{F}_\msf{UBC},$ $\mc{G}_\msf{clock})$-hybrid model, where $\msf{leak}(\Cl)=\Cl+2$ and $\msf{delay}=3$.
\end{corollary}
\NL{Both realization are under the same instance, e.g same leak function and delay. Thus, the composition theorem holds.}}

%% file: SBC.tex
\section{Simultaneous Broadcast}\label{sec:SBC}

In this section, we present our formal study of the simultaneous broadcast (SBC) notion in the UC framework, which comprises a new functionality $\mc{F}_\msf{SBC}$ and a TLE-based construction that we prove it UC-realizes $\mc{F}_\msf{SBC}$. Our approach revisits and improves upon the work of Hevia~\cite{Hevia06} w.r.t. several aspects. In particular,
\begin{enumerate}[leftmargin=*]
    \item We consider SBC executions where the honest parties agree on a well-defined \emph{broadcast period}, outside of which all broadcast messages are ignored. We argue that this setting is plausible: recall that simultaneity suggests that no sender broadcasts a message depending on the messages broadcast by other parties. If there is no such broadcast period, then liveness and simultaneity are in conflict, in the sense that a malicious sender could wait indefinitely until all honest parties are forced to either (i) abort, or (ii) reveal their messages before all (malicious) senders broadcast their values. On the contrary, within an agreed valid period, honest parties can safely broadcast knowing that every invalid message will be discarded. Moreover, unlike~\cite{Hevia06}, full participation of all parties is not necessary for the termination of the protocol execution. In Section~\ref{sec:applications}, we propose practical use cases where our SBC setting is greatly desired.
    \item The SBC functionality of~\cite{Hevia06} is designed w.r.t. the synchronous communication setting in~\cite{UC}. As shown in~\cite{KatzMTZ13}, this setting has limitations (specifically, guarantee of termination) that are lifted when using $\mc{G}_\msf{clock}$. In our formal treatment, synchronicity is captured in the state-of-the-art $\mc{G}_\msf{clock}$-hybrid model.
    \item The SBC construction in~\cite{Hevia06} is proven secure only against adversaries that corrupt a minority of all parties. By utilizing TLE, our work introduces the first SBC protocol that is UC secure against any adversary corrupting $t<n$ parties.
    \end{enumerate}
\medskip
\begin{figure}[H]
\begin{tcolorbox}[enhanced, colback=white, arc=10pt, drop shadow southeast]
\noindent\emph{\underline{The simultaneous broadcast functionality $\mc{F}^{\Phi,\Delta,\alpha}_\msf{SBC}(\mbf{P})$.}}\\[5pt]
\begin{small}
The functionality initializes the list $L_\msf{pend}$ of pending messages as empty and two variables $t_\msf{start},t_\msf{end}$ to $\bot$. It also maintains the set of corrupted parties, $\mathbf{P}_\msf{corr}$, initialized as empty.\\[2pt]
\extitem Upon receiving $(\sid,\textsc{Broadcast},M)$ from $P\in\mbf{P}\setminus\mbf{P}_\msf{corr}$ or $(\sid,\textsc{Broadcast},M,P)$ from $\mc{S}$ on behalf of $P\in\mbf{P}_\msf{corr}$, it does:
\begin{enumerate}[leftmargin=*]
    \item It reads the time $\Cl$ from $\mc{G}_\msf{clock}$.
    \item If $t_\msf{start}=\bot$, it sets $t_\msf{start}\leftarrow\Cl$ and $t_\msf{end}\leftarrow t_\msf{start}+\Phi$.
    \item If $t_\msf{start}\leq\Cl< t_\msf{end}$, it does :
    \begin{enumerate}
    \item It picks a unique random $\msf{tag}$ from $\{0,1\}^\lambda$. 
    \item If $P\in\mbf{P}\setminus\mbf{P}_\msf{corr}$, it adds $(\msf{tag},M,P,\Cl,0)$ to $L_\msf{pend}$ and sends $(\sid,\textsc{Sender},\msf{tag},0^{|M|},P)$ to $\mc{S}$. Otherwise, it adds $(\msf{tag},M,P,\Cl,1)$ to $L_\msf{pend}$ and sends $(\sid,\textsc{Sender},\msf{tag},M,P)$ to $\mc{S}$.
    \end{enumerate}
\end{enumerate}
\vspace{2pt}
\extitem Upon receiving $(\sid,\textsc{Corruption\_Request})$ from $\mc{S}$, it sends $(\sid,\textsc{Corruption\_Request},\langle(\msf{tag},M,P,\Cl^*,0)\in L_\msf{pend}: P\in\mbf{P}_\msf{corr}\rangle)$ to $\mc{S}$.\\[2pt]
\extitem Upon receiving $(\sid,\textsc{Allow},\msf{tag},\tilde{M},\tilde{P})$ from $\mc{S}$, it does:
\begin{enumerate}[leftmargin=*]
    \item It reads the time $\Cl$ from $\mc{G}_\msf{clock}$.
    \item If $t_\msf{start}\leq\Cl< t_\msf{end}$ and there is a tuple $(\msf{tag},M,\tilde{P},\Cl^*,0)\in L_\msf{pend}$ and $\tilde{P}\in\mbf{P}_\msf{corr}$, it updates the tuple as $(\msf{tag},\tilde{M},\tilde{P},\Cl^*,1)$ and sends $(\sid,\textsc{Allow\_OK})$ to $\mc{S}$. Otherwise, it ignores the message.
\end{enumerate}
\vspace{2pt}
\extitem Upon receiving $(\sid_C,\textsc{Advance\_Clock})$ from $P\in\mbf{P}\setminus\mbf{P}_\msf{corr}$, it does:
\begin{enumerate}[leftmargin=*]
    \item It reads the time $\Cl$ from $\mc{G}_\msf{clock}$. 
    \item If this is the first time it has received a $(\sid_C,\textsc{Advance\_Clock})$ message from $P$ during round $\Cl$, then
    \begin{enumerate}
        \item If it has received no other $(\sid_C,\textsc{Advance\_Clock})$ message during round $\Cl$,
        \begin{enumerate}
            \item If $\Cl=t_\msf{end}$, then it does:
            \begin{enumerate}
                \item It updates every tuple $(\cdot,\cdot,P^*,\cdot,0)\in L_\msf{pend}$ such that $P^*\in\mbf{P}\setminus\mbf{P}_\msf{corr}$ as $(\cdot,\cdot,P^*,\cdot,1)$ (to guarantee the broadcast of messages from always honest parties).
                \item It sorts $L_\msf{pend}$ lexicographically according to the second coordinate (messages).
            \end{enumerate}
            \item If $\Cl=t_\msf{end}+\Delta-\alpha$, it sends $(\sid,\textsc{Broadcast},\langle (\msf{tag},M)\rangle_{(\msf{tag},M,\cdot,\cdot,1)\in L_\msf{pend}})$ to $\mc{S}$.
        \end{enumerate}
       \item If $\Cl=t_\msf{end}+\Delta$, it sends $(\sid,\textsc{Broadcast},\langle M\rangle_{(\cdot,M,\cdot,\cdot,1)\in L_\msf{pend}})$ to $P$.
    \end{enumerate}
       \item It returns $(\sid_C, \textsc{Advance\_Clock})$ to $P$ with destination identity $\mc{G}_\msf{clock}$.
\end{enumerate}
\end{small}
\end{tcolorbox}
\caption{The functionality $\mc{F}^{\Phi,\Delta,\alpha}_\msf{SBC}$ interacting with the parties in $\mbf{P}$ and the simulator $\mc{S}$, parameterized by time span $\Phi$, delay $\Delta$, and simulator advantage $\alpha$.}
\label{fig:SBC}
\end{figure}
    \noindent\textbf{The simultaneous broadcast (SBC) functionality.} Our SBC functionality $\mc{F}_\msf{SBC}$ (cf. Figure~\ref{fig:SBC}) interacts with $\mc{G}_\msf{clock}$ and is parameterized by a broadcast time span $\Phi$, a message delivery delay $\Delta$ and a simulator advantage $\alpha$. Upon receiving the first \textsc{Broadcast} request, it sets the current global time as the beginning of the broadcast period that lasts $\Phi$ rounds. If a \textsc{Broadcast} request was made by an honest sender $P$, then the functionality leaks only the sender's identity. Besides, all \textsc{Broadcast} requests are recorded as long as they are made within the broadcast period. The recorded messages are issued to each party (resp. the simulator) $\Delta$ rounds (resp. $\Delta-\alpha$ rounds) after the end of the broadcast period.

\begin{figure}[H]
\begin{tcolorbox}[enhanced, colback=white, arc=10pt, drop shadow southeast]
\noindent\emph{\underline{The SBC protocol $\Pi_\msf{SBC}(\mc{F}_\msf{UBC},\mc{F}^{\msf{leak},\msf{delay}}_\msf{TLE},\mc{F}_\msf{RO},\Phi,\Delta,\mbf{P})$.}}\\[5pt]
\begin{small}
Every party $P$ maintains a list $L^P_\msf{pend}$ of messages under pending encryption and a list $L^P_\msf{rec}$ of received ciphertexts both initialized as empty, and four variables $t^P_\msf{awake}$, $t^P_\msf{end}$, $\tau^P_\msf{rel}$, $\msf{first}^P$, all initialized to $\bot$. All parties understand a special message `$\mathtt{Wake\_Up}$' that is not in the broadcast message space.\\[2pt]
\extitem Upon receiving $(\sid,\textsc{Broadcast},M)$ from $\mc{Z}$, the party $P$ does:
\begin{enumerate}[leftmargin=*]
    \item If $t^P_\msf{awake}=\bot$, she sets $\msf{first}^P\leftarrow M$ and sends $(\sid,\textsc{Broadcast},\mathtt{Wake\_Up})$ to $\mc{F}_\msf{UBC}$.
    \item If $t^P_\msf{awake}\neq\bot$, she does:
    \begin{enumerate}
        \item She reads the time $\Cl$ from $\mc{G}_\msf{clock}$.
        \item If $\Cl\geq t^P_\msf{end}-\msf{delay}$, she ignores the message\footnote{The reason is that due to TLE ciphertext generation time ($\msf{delay}$ rounds), if $\Cl\geq t^P_\msf{end}-\msf{delay}$, then the message would not be ready for broadcast before $t^P_\msf{end}$.}. 
        \item She chooses a randomness $\msf{\rho}\overset{\$}{\leftarrow}\{0,1\}^\lambda$.
        \item She adds $(\rho,M)$ in $L^P_\msf{pend}$.
    \item\label{item:tle_end} She sends $(\sid,\textsc{Enc},\rho,\tau^P_\msf{rel})$ to $\mc{F}^{\msf{leak},\msf{delay}}_\msf{TLE}$.
    \end{enumerate}
\end{enumerate}
\vspace{2pt}
\extitem Upon receiving $(\sid,\textsc{Broadcast},\mathtt{Wake\_Up})$ from $\mc{F}_\msf{UBC}$, if $t^P_\msf{awake}=\bot$, the party $P$ does:
\begin{enumerate}[leftmargin=*]
    \item She reads the time $\Cl$ from $\mc{G}_\msf{clock}$.
    \item She sets $t^P_\msf{awake}\leftarrow\Cl$, $t^P_\msf{end}\leftarrow t^P_\msf{awake}+\Phi$, and  $\tau^P_\msf{rel}\leftarrow t^P_\msf{end}+\Delta$ (i.e., all parties agree on the start and end of the broadcast period, as well as the time-lock decryption time).
    \item If $\msf{first}^P\neq\bot$, she parses the (unique) pair in $L^P_\msf{pend}$ that contains $\msf{first}^P$ as $(\rho,\msf{first}^P)$. Then, she sends $(\sid,\textsc{Enc},\rho,\tau^P_\msf{rel})$ to $\mc{F}^{\msf{leak},\msf{delay}}_\msf{TLE}$ (this check is true only if $P$ broadcasts her first message when acting as the first sender in the session).
\end{enumerate}
\extitem Upon receiving $(\sid,\textsc{Broadcast},(c^*,\tau^*,y^*))$ from $\mc{F}_\msf{UBC}$, if $\tau^*=\tau^P_\msf{rel}$ and for every $(c',y')\in L^P_\msf{rec}:c'\neq c^*\wedge y'\neq y^*$, then the party $P$ adds $(c^*,y^*)$ to $L^P_\msf{rec}$.\\[2pt]
\extitem Upon receiving $(\msf{sid}_C,\textsc{Advance\_Clock})$ by $\mc{Z}$, the party $P$ does:
\begin{enumerate}[leftmargin=*]
    \item She reads the time $\Cl$ from $\mc{G}_\msf{clock}$. If this is not the first time she has received a $(\msf{sid}_C,\textsc{Advance\_Clock})$ command during round $\Cl$, she ignores the message.
    \item If $t^P_\msf{awake}\leq\Cl<t^P_\msf{end}$, she sends $(\sid,\textsc{Retrieve})$ to $\mc{F}^{\msf{leak},\msf{delay}}_\msf{TLE}$ to obtain the encryptions of messages that she requested $\msf{delay}$ rounds earlier. Upon receiving $(\sid,\textsc{Encrypted},T)$ from $\mc{F}^{\msf{leak},\msf{delay}}_\msf{TLE}$, she does:
    \begin{enumerate}
        \item She parses $T$ as a list of tuples of the form $(\rho,c,\tau^P_\msf{rel})$.
    \item For every $(\rho,c,\tau^P_\msf{rel})\in T$ such that there is a pair $(\rho,M)\in L^P_\msf{pend}$, she does:
    \begin{enumerate}
        \item She queries $\mc{F}_\msf{RO}$ on $\rho$ and receives a response $\eta$.
        \item She computes $y\leftarrow M\oplus\eta$.
        \item She sends $(\sid,\textsc{Broadcast},(c,\tau^P_\msf{rel},y))$ to $\mc{F}_\msf{UBC}$.
    \end{enumerate}
    \end{enumerate} 
    \item If $\Cl=\tau^P_\msf{rel}$, then for every $(c^*,y^*)\in L_\msf{rec}^P$, she does:
    \begin{enumerate}
         \item She sends $(\sid,\textsc{Dec},c^*,\tau^P_\msf{rel})$ to $\mc{F}^{\msf{leak},\msf{delay}}_\msf{TLE}$. Upon receiving $(\sid,\textsc{Dec},c^*,\tau^P_\msf{rel},\rho^*)$ from $\mc{F}^{\msf{leak},\msf{delay}}_\msf{TLE}$, if $\rho^*\notin\{\bot,\textsc{More\_Time},\textsc{Invalid\_Time}\}$, she queries $\mc{F}_\msf{RO}$ on $\rho^*$ and receives a response $\eta^*$.
    \item She computes $M^*\leftarrow y^*\oplus\eta^*$.
    \item She sends $(\sid,\textsc{Broadcast},M^*)$ to $\mc{Z}$.
    \end{enumerate}
    \item She sends $(\msf{sid}_C,\textsc{Advance\_Clock})$ to $\mc{F}_\msf{UBC}$. Upon receiving $(\msf{sid}_C,\textsc{Advance\_Clock})$ from $\mc{F}_\msf{UBC}$, she forwards  $(\msf{sid}_C,\textsc{Advance\_Clock})$ to $\mc{G}_\msf{clock}$ and completes her round.
\end{enumerate}
\end{small}

\end{tcolorbox}
\caption{The protocol $\Pi_\msf{SBC}$ with the parties in $\mbf{P}$.}
\label{fig:real_SBC}
\end{figure}

\noindent\textbf{The SBC protocol.} Our SBC protocol (cf. Figure~\ref{fig:real_SBC}) is over $\mc{F}_\msf{UBC}$ and deploys $\mc{F}^{\msf{leak},\msf{delay}}_\msf{TLE}$ to achieve simultaneity, and $\mc{F}_\msf{RO}$ for equivocation. In the beginning, the first sender notifies via $\mc{F}_\msf{UBC}$ the other parties of the start of the broadcast period via a special `\texttt{Wake\_Up}' message. By the properties of UBC, all honest parties agree on the time frame of the broadcast period that lasts $\Phi$ rounds. During the broadcast period, in order to broadcast a message $M$, the sender chooses a randomness $\rho$ and interacts with $\mc{F}^{\msf{leak},\msf{delay}}_\msf{TLE}$ to obtain a TLE ciphertext $c$ of $\rho$ (after $\msf{delay}$ rounds). By default, $c$ is set to be decrypted $\Delta$ rounds after the end of the broadcast period. Then, she makes an RO query for $\rho$, receives a response $\eta$ and broadcasts $c$ and $M\oplus\eta$ via $\mc{F}_\msf{UBC}$. Any recipient of $c,M\oplus\eta$ can retrieve the message $\Delta$ rounds after the end of the broadcast period by (i) obtaining $\rho$ via a decryption request of $c$ to $\mc{F}^{\msf{leak},\msf{delay}}_\msf{TLE}$, (ii) obtaining $\eta$ as a RO response to query $\rho$, and (iii)
computing $M\leftarrow (M\oplus\eta)\oplus\eta$.

We prove the security of our SBC construction below.

\begin{theorem}\label{thm:SBC}
Let $\msf{leak}(\cdot),\msf{delay}$ be the leakage and delay parameters of $\mc{F}_\msf{TLE}$. Let $\Phi,\Delta$ be positive integers such that $\Phi>\msf{delay}$ and $\Delta>\underset{\Cl^*}{\mathrm{max}}\{\msf{leak}(\Cl^*)-\Cl^*\}$. The protocol $\Pi_\msf{SBC}$ in Figure~\ref{fig:real_SBC} UC-realizes $\mc{F}^{\Phi,\Delta,\alpha}_\msf{SBC}$ in the $(\mc{F}_\msf{UBC},\mc{F}^{\msf{leak},\msf{delay}}_\msf{TLE},\mc{F}_\msf{RO},\mc{G}_\msf{clock})$-hybrid model against an adaptive adversary corrupting $t<n$ parties, where the simulator advantage is $\alpha=\underset{\Cl^*}{\mathrm{max}}\{\msf{leak}(\Cl^*)-\Cl^*\}+1$.
\end{theorem}

\begin{proof}
We devise a simulator $\mc{S}_\msf{SBC}$ that given a real-world adversary $\mc{A}$, simulates an execution of $\Pi_\msf{SBC}$ in the $(\mc{F}_\msf{UBC},\mc{F}^{\msf{leak},\msf{delay}}_\msf{TLE},\mc{F}_\msf{RO},\mc{G}_\msf{clock})$-hybrid world in the presence of $\mc{A}$. $\mc{S}_\msf{SBC}$ acts as a proxy forwarding messages from/to the environment $\mc{Z}$ to/from $\mc{A}$. It also simulates the behavior of honest parties, and controls the local functionalities $\mc{F}_\msf{UBC}$, $\mc{F}^{\msf{leak},\msf{delay}}_\msf{TLE}$, and $\mc{F}_\msf{RO}$. 

In particular, $\mc{S}_\msf{SBC}$ initializes two lists $L_\msf{pend}$, $L_\msf{ubc}$  as empty and three variables $t_\msf{awake}$, $t_\msf{end}$, $\tau_\msf{rel}$ to $\bot$, and operates as follows.\\[2pt]
\extitem Whenever $\mc{A}$ corrupts a simulated party $P$, $\mc{S}_\msf{SBC}$ corrupts the corresponding dummy party. It also sends $(\sid,\textsc{Corruption\_Request})$ to $\mc{F}^{\Phi,\Delta,\alpha}_\msf{SBC}$ and records the functionality's response that is a sequence of tuples of the form $(\msf{tag}^*,M^*,P^*,\Cl^*,0)$, where $P^*$ is corrupted in the ideal world.\\[2pt]
\extitem Upon receiving $(\sid,\textsc{Sender},\msf{tag},0^{|M|},P)$ from $\mc{F}^{\Phi,\Delta,\alpha}_\msf{SBC}$, where $P$ is honest, it does:
\begin{enumerate}
    \item If this is the first time it has received such a message referring to $P$ as a sender and $t_\msf{awake}=\bot$, it simulates a broadcast of the special $\mathtt{Wake\_Up}$ message via $\mc{F}_\msf{UBC}$ by adding $(\msf{tag}_0,\mathtt{Wake\_Up},P)$ to $L_\msf{ubc}$ and providing $(\sid,\textsc{Broadcast},\msf{tag}_0,\mathtt{Wake\_Up},P)$ as leakage of $\mc{F}_\msf{UBC}$ to $\mc{A}$, where $\msf{tag}_0$ is a unique random string in $\{0,1\}^\lambda$. 
    \begin{itemize}
        \item Upon receiving as $\mc{F}_\msf{UBC}$ a message $(\sid,\textsc{Allow},\msf{tag}_0,M_0)$ from $\mc{A}$, if, meanwhile, $P$ has been corrupted, it simulates the broadcast of $M_0$ via $\mc{F}_\msf{UBC}$ by (a) sending $(\sid,\textsc{Broadcast},$ $M_0)$ to all simulated parties and $(\sid,\textsc{Broadcast},M_0,P)$ to $\mc{A}$, and (b) deleting $(\msf{tag}_0,\cdot,P)$ from $L_\msf{ubc}$. In addition, if $M_0=\mathtt{Wake\_Up}$ and $t_\msf{awake}=\bot$, it specifies the broadcast period (consistently to $\mc{F}^{\Phi,\Delta,\alpha}_\msf{SBC}$), as well as the time-lock decryption time, for the simulated session. Namely, it reads the time $\Cl$ from $\mc{G}_\msf{clock}$ and sets $t_\msf{awake}\leftarrow\Cl$, $t_\msf{end}\leftarrow t_\msf{awake}+\Phi$ and $\tau_\msf{rel}\leftarrow t_\msf{awake}+\Delta$.
    \end{itemize}
    \item It reads the time $\Cl$ from $\mc{G}_\msf{clock}$. If $t_\msf{awake}=\bot$ or $\Cl\geq t_\msf{end}-\msf{delay}$, it takes no further action, just as a real-world honest party would ignore a message that cannot be broadcast within the broadcast period. Otherwise, it chooses a randomness $\msf{\rho}\overset{\$}{\leftarrow}\{0,1\}^\lambda$ and a random value $y$ from the image of $\mc{F}_\msf{RO}$ and simulates an encryption of $\rho$ via $\mc{F}^{\msf{leak},\msf{delay}}_\msf{TLE}$ as follows:
    \begin{enumerate}
        \item It uses $\msf{tag}$ for the record of $\mc{F}^{\msf{leak},\msf{delay}}_\msf{TLE}$ and adds the tuple $(\rho,\msf{Null},\msf{tag},\Cl,y,P)$ to the list $L_\msf{pend}$. Note that by the description of $\mc{F}_\msf{SBC}$, it holds that $\Cl\geq t_\msf{awake}$, so overall, $t_\msf{awake}\leq\Cl<t_\msf{end}-\msf{delay}$ holds.
        \item It sends $(\sid,\textsc{Enc},\tau_\msf{rel},\msf{tag},\Cl,0^{|M|})$ as leakage of $\mc{F}^{\msf{leak},\msf{delay}}_\msf{TLE}$ to $\mc{A}$. Upon receiving the token back from $\mc{A}$, it returns $(\sid,\textsc{Encrypting})$ to $P$.
    \end{enumerate}
\end{enumerate}
\vspace{2pt}
\extitem Upon receiving as $\mc{F}^{\msf{leak},\msf{delay}}_\msf{TLE}$ a message $(\sid,\textsc{Update},\mbf{C})$ from $\mc{A}$, it parses $\mbf{C}$ as a set of pairs of  ciphertexts and tags. For every pair $(c,\msf{tag})\in\mbf{C}$ such that a tuple $(\rho,\msf{Null},\msf{tag},\Cl^*,y,P)$ is recorded in $L_\msf{pend}$, it updates the recorded tuple as $(\rho,c,\msf{tag},\Cl^*,y,P)$.\\[2pt]
\extitem Upon receiving as $\mc{F}^{\msf{leak},\msf{delay}}_\msf{TLE}$ a message $(\sid,\textsc{Leakage})$ from $\mc{A}$, it does:
\begin{enumerate}
    \item It reads the time $\Cl$ from $\mc{G}_\msf{clock}$.
    \item If $\tau_\msf{rel}\leq\msf{leak}(\Cl)$, it creates a set $\mbf{L}_\Cl$ that contains all triples $(\rho^*,c^*,\tau_\msf{rel})$ such that there is a tuple $(\rho^*,c^*,\msf{tag}^*,\Cl^*,y^*,P^*)\in L_\msf{pend}$. Note that by the definition of $\alpha$ and $\tau_\msf{rel}$, we have that 
    \begin{equation}\label{eq:leakage_time}
    \tau_\msf{rel}\leq\msf{leak}(\Cl)\Rightarrow t_\msf{end}+\Delta\leq\Cl+\alpha-1\Rightarrow t_\msf{end}+\Delta-\alpha<\Cl.
    \end{equation}
    Otherwise, it sets $\mbf{L}_\Cl\leftarrow\emptyset$.
    \item It sends $(\sid,\textsc{Leakage},\mbf{L}_\Cl)$ to $\mc{A}$.
\end{enumerate}
\vspace{2pt}
\extitem Upon receiving as $\mc{F}^{\msf{leak},\msf{delay}}_\msf{TLE}$ a message $(\sid,\textsc{Retrieve})$ from a corrupted party $P$, it does:
\begin{enumerate}
    \item It reads the time $\Cl$ from $\mc{G}_\msf{clock}$.
    \item It creates a set $\mbf{E}_{\Cl,P}$ that contains all triples $(\rho^*,c^*,\tau_\msf{rel})$ such that there is a tuple $(\rho^*,c^*,\msf{tag}^*,$ $\Cl^*,y^*,P)\in L_\msf{pend}$, where (i) $c^*\neq\msf{Null}$, and (ii) $\Cl-\Cl^*\geq\msf{delay}$.
    \item It sends $(\sid,\textsc{Encrypted},\mbf{E}_{\Cl,P})$ to $P$.
\end{enumerate}
\vspace{2pt}
\extitem Upon receiving $(\sid,\textsc{Broadcast},BC)$ from $\mc{F}^{\Phi,\Delta,\alpha}_\msf{SBC}$, it does:
\begin{enumerate}
    \item It reads the time $\Cl$ from $\mc{G}_\msf{clock}$.
    \item If $\Cl\neq t_\msf{end}+\Delta-\alpha$, it aborts simulation. Otherwise, it parses $BC$ as a list of pairs of tags and messages of the form $(\msf{tag}^*,M^*)$. For every pair $(\msf{tag}^*,M^*)\in BC$, it does:
    \begin{enumerate}
        \item It finds the corresponding tuple $(\rho^*,c^*,\msf{tag}^*,\Cl^*,y^*,P^*)\in L_\msf{pend}$.
        \item It adds $(\msf{tag}^*,M^*,P^*)$ to a list $L^{P^*
        }_\msf{bc}$ (initialized as empty).
        \item It programs $\mc{F}_\msf{RO}$ on $\rho^*$ as $(\rho^*,M^*\oplus y^*)$. If $\rho^*$ has been made as query by $\mc{A}$ earlier, it aborts simulation.
    \end{enumerate}
\end{enumerate}
\vspace{2pt}
\extitem Upon receiving $(\sid_C,\textsc{Advance\_Clock},P)$ from $\mc{G}_\msf{clock}$ for $P$ honest, it does:
\begin{enumerate}
    \item It reads the time $\Cl$ from $\mc{G}_\msf{clock}$.
    \item If $t_\msf{awake}\leq\Cl<t_\msf{end}$, it does:
    \begin{enumerate}
    \item It  creates a list $T_{\Cl,P}$ that includes the tuples of the form $(\rho,c,\msf{tag},y)$ that derive from tuples $(\rho,c,\msf{tag},\Cl-\msf{delay},y,P)$ recorded in $L_\msf{pend}$ during round $\Cl-\msf{delay}$.
    \item For every $(\rho,c,\msf{tag},y)\in T_{\Cl,P}$, it does:
    \begin{enumerate}
        \item It simulates broadcasting of $(c,\tau_\msf{rel},y)$ via $\mc{F}_\msf{UBC}$ by adding $(\msf{tag},(c,\tau_\msf{rel},y),P)$ to $L_\msf{ubc}$ and sending $(\sid,\textsc{Broadcast},\msf{tag},(c,\tau_\msf{rel},y),P)$ as leakage of $\mc{F}_\msf{UBC}$ to $\mc{A}$.
         \begin{itemize}
             \item Upon receiving $(\sid,\textsc{Allow},\msf{tag},\tilde{C})$ from $\mc{A}$, if, meanwhile, $P$ has been corrupted, it simulates the broadcast of $\tilde{C}$ via $\mc{F}_\msf{UBC}$ by (a) sending $(\sid,\textsc{Broadcast},$ $\tilde{C})$ to all simulated parties and $(\sid,\textsc{Broadcast},\tilde{C},P)$ to $\mc{A}$, and (b) deleting $(\msf{tag},\tilde{C},P)$ from $L_\msf{ubc}$. In addition, if $\tilde{C}$ can be parsed as a triple $(\tilde{c},\tilde{\tau},\tilde{y})$ such that $\tilde{c}$, $\tilde{y}$ are in correct format and $\tilde{\tau}=\tau_\msf{rel}$, it ``decrypts'' $\tilde{c}$ as follows:
             \begin{itemize}
                 \item If $\tilde{c}$ or $\tilde{y}$ have been broadcast previously, then it ignores $(\tilde{c},\tilde{\tau},\tilde{y})$ just as real-world honest parties do when they receive ``replayed'' messages.
                  \item  If there is a tuple $(\tilde{\rho},\tilde{c},\tilde{\msf{tag}},\tilde{\Cl},\tilde{y},P)\in L_\msf{pend}$, then it marks $\tilde{\rho}$ as the decryption of $\tilde{c}$ and aborts simulation if $\tilde{\rho}$ has been made as query by $\mc{A}$ before corrupting $P$. In case of no abort, recall that $\mc{S}_\msf{SBC}$ has already retrieved a corresponding tuple $(\tilde{\msf{tag}},\tilde{M},P,\tilde{\Cl},0)$ via a $\textsc{Corruption\_Request}$ message to $\mc{F}^{\Phi,\Delta,\alpha}_\msf{SBC}$. Hence, it programs $\mc{F}_\msf{RO}$ as $(\tilde{\rho},\tilde{M}\oplus\tilde{y})$. Then, it sends $(\sid,\textsc{Allow},\msf{tag},\tilde{M})$ to $\mc{F}^{\Phi,\Delta,\alpha}_\msf{SBC}$ and receives the response $(\sid,\textsc{Allow\_OK})$.
                  \item  If such a tuple does not exist, then by playing the role of $\mc{F}^{\msf{leak},\msf{delay}}_\msf{TLE}$, it requests the decryption of $\tilde{c}$ from $\mc{A}$ (i.e., it sends $(\sid,\textsc{Dec},\tilde{c},\tilde{\tau})$ to $\mc{A}$) and records its response as $(\tilde{\rho},\tilde{c})$. In addition, it extracts the adversarial message $\tilde{M}\leftarrow\tilde{y}\oplus\tilde{\eta}$, where $\tilde{\eta}$ is (i) the existing programming of $\mc{F}_\msf{RO}$ on $\tilde{\rho}$, if $\tilde{\rho}$ has been made as query by $\mc{A}$ earlier, or (ii) the fresh programming of $\mc{F}_\msf{RO}$ on $\tilde{\rho}$, otherwise. Then, it sends $(\sid,\textsc{Allow},\msf{tag},\tilde{M})$ to $\mc{F}^{\Phi,\Delta,\alpha}_\msf{SBC}$ and receives the response $(\sid,\textsc{Allow\_OK})$. 
        
             \end{itemize}       
         \end{itemize}
\item If $P$ remains honest, it operates as $\mc{F}_\msf{UBC}$ on message $(\sid_C,\textsc{Advance\_Clock})$. Namely, for every $(\rho,c,\msf{tag},y)\in T_{\Cl,P}$, it simulates broadcasting of $(c,\tau_\msf{rel},y)$ via $\mc{F}_\msf{UBC}$ by sending $(\sid,\textsc{Broadcast},(c,\tau_\msf{rel},y))$ to all parties and $(\sid,\textsc{Broadcast},$ $(c,\tau_\msf{rel},y),P)$ to $\mc{A}$.
\end{enumerate}
  \item If $t_\msf{awake}=\bot$ and $(\msf{tag}_0,\mathtt{Wake\_Up},P)\in L_\msf{ubc}$, it simulates the broadcast of $\mathtt{Wake\_Up}$ via $\mc{F}_\msf{UBC}$ by (a) sending $(\sid,\textsc{Broadcast},\mathtt{Wake\_Up})$ to all parties and $(\sid,\textsc{Broadcast},$ $\mathtt{Wake\_Up},P)$ to $\mc{A}$, and (b) deleting $(\msf{tag}_0,\mathtt{Wake\_Up},P)$ from $L_\msf{ubc}$. Then, it reads the time $\Cl$ from $\mc{G}_\msf{clock}$ and sets $t_\msf{awake}\leftarrow\Cl$, $t_\msf{end}\leftarrow t_\msf{awake}+\Phi$ and $\tau_\msf{rel}\leftarrow t_\msf{awake}+\Delta$.

    \end{enumerate}
\end{enumerate}
\vspace{2pt}
\extitem Upon receiving as $\mc{F}_\msf{UBC}$ a message $(\sid,\textsc{Broadcast},C,P)$ from $\mc{A}$ on behalf of a corrupted party $P$, it simulates the broadcast of $C$ via $\mc{F}_\msf{UBC}$ by sending $(\sid,\textsc{Broadcast},C)$ to all parties and $(\sid,\textsc{Broadcast},C,P)$ to $\mc{A}$. Then, it reads the time $\Cl$ from $\mc{G}_\msf{clock}$. If $t_\msf{awake}\leq\Cl<t_\msf{end}$ and $C$ can be parsed as a triple $(c,\tau,y)$ such that $c$, $y$ are in correct format and $\tau=\tau_\msf{rel}$, it does:
\begin{enumerate}
    \item If $c$ or $y$ have been broadcast again previously, then it ignores $(c,\tau,y)$ just as real-world honest parties do when they receive ``replayed'' messages.
    \item By playing the role of $\mc{F}^{\msf{leak},\msf{delay}}_\msf{TLE}$, it requests the decryption of $c$ from $\mc{A}$ (i.e., it sends $(\sid,\textsc{Dec},c,\tau)$ to $\mc{A}$) and records its response as $(\rho,c)$. In addition, it extracts the adversarial message $M\leftarrow y\oplus\eta$, where $\eta$ is (i) the existing programming of $\mc{F}_\msf{RO}$ on $\rho$, if $\rho$ has been made as query by $\mc{A}$ earlier, or (ii) the fresh programming of $\mc{F}_\msf{RO}$ on $\rho$, otherwise. Then, it sends $(\sid,\textsc{Allow},\msf{tag},M)$ to $\mc{F}^{\Phi,\Delta,\alpha}_\msf{SBC}$ and receives the response $(\sid,\textsc{Allow\_OK})$. 
\end{enumerate}

\vspace{5pt}
\indent For the simulation above, we define the following events: let $E$ be the event that some randomness $\rho$ or some $\msf{tag}$ is used more than once. Let $A$ be the event that simulation aborts. We observe that 
\begin{itemize}
    \item Since $\rho$ and $\msf{tag}$ are randomly chosen from a space of exponential size by $\mc{S}_\msf{SBC}$ and $\mc{F}^{\Phi,\Delta,\alpha}_\msf{SBC}$, respectively, it holds that $\Pr[E]=\msf{negl}(\lambda)$.
    \item By Eq.~\eqref{eq:leakage_time} and the synchronicity between the functionality's $t_\msf{start}$ and the simulator's $t_\msf{awake}$ that $\mc{G}_\msf{clock}$ supports, we have that $\mc{S}_\msf{SBC}$ will provide $\mc{A}$ with the leakage from $\mc{F}^{\msf{leak},\msf{delay}}_\msf{TLE}$ at least one round later than the time $t_\msf{end}+\Delta-\alpha$ that $\mc{S}_\msf{SBC}$ receives the list $BC$ of broadcast messages from $\mc{F}^{\Phi,\Delta,\alpha}_\msf{SBC}$. Thus, the probability that $\mc{A}$ guesses a decryption $\rho$ of come ciphertext $c$ before getting the leakage (so that $\mc{S}_\msf{SBC}$ cannot equivocate accordingly and abort) is negligible. Namely, $\Pr[A]=\msf{negl}(\lambda)$.
\end{itemize}

By the above observations, except with some $\msf{negl}(\lambda)$ probability, the events $E$ or $A$ do not happen, i.e., every randomness and tag are unique, and the simulation is completed. In such case, $\mc{S}_\msf{SBC}$ simulates a protocol execution by setting a common start time $t_\msf{awake}$, end time $t_\msf{end}$, and time-lock decryption time $\tau_\msf{rel}$ for all honest parties. The latter is consistent with a real-world execution because by the broadcast property of $\mc{F}_\msf{UBC}$ all real-world honest parties will agree on the start time, end time $t_\msf{end}$, and time-lock decryption time of the session, i.e. for every $P,P'$: $t^P_\msf{awake}=t^{P'}_\msf{awake}$, $t^P_\msf{end}=t^{P'}_\msf{end}$, and  $\tau^P_\msf{rel}=\tau^{P'}_\msf{rel}$.

Therefore, and given that by equivocating, $\mc{S}$ perfectly emulates an $\mc{F}_\msf{RO}$ instantiation, we have that, except with some $\msf{negl}(\lambda)$ probability, the simulation is perfect. This completes the proof.  
\end{proof}
From Lemma~\ref{lem:real_UBC}, Lemma~\ref{lem:real_FBC},  Theorem~\ref{thm:real_TLE}, and Theorem~\ref{thm:SBC}, we directly get the following corollary.
\begin{corollary}\label{cor:TLE}
There exists a protocol that UC-realises $\mc{F}^{\Phi,\Delta,\alpha}_\msf{SBC}$ in the $(\mc{F}_\msf{cert},\mc{W}_q(\mc{F}^*_\msf{RO}),\mc{F}_\msf{RO},$ $\mc{G}_\msf{clock})$-hybrid model, where $\Phi>3$, $\Delta>2$, and $\alpha=3$. \end{corollary}

%% file: applications.tex
\section{Applications of simultaneous broadcast}\label{sec:applications}
\subsection{Distributed uniform random string generation}\label{sec:applications_urs}
\noindent
\textbf{The delayed uniform random string (DURS) functionality.}
The DURS functionality is along the lines of the common reference string (CRS) functionality in~\cite{UC}. 
The functionality draws a single random string $r$ uniformly at random, and delivers $r$ upon request. The delivery of $r$ is delayed, in the sense that the party who made an early request has to wait until $\Delta$ time has elapsed since the first request was made. Besides, the simulator has an advantage $\alpha$, i.e., it can obtain $r$ (on behalf of some corrupted party) when $\Delta-\alpha$ time has elapsed since the first request was made. The DURS functionality is presented in detail in Figure~\ref{fig:urs}.

\begin{figure}[H]
\begin{tcolorbox}[enhanced, colback=white, arc=10pt, drop shadow southeast]

\noindent\emph{\underline{The delayed uniform reference string functionality $\mathcal{F}^{\Delta,\alpha}_\msf{DURS}(\mbf{P})$.}}\\[5pt]
\begin{small}
The functionality initializes two variables $\msf{urs}$, $t_\msf{start}$ as $\bot$ and a flag $f^P_\msf{wait}$ to $0$, for every party $P\in\mbf{P}$. It also maintains the set of corrupted parties, $\mathbf{P}_\msf{corr}$, initialized as empty.\\[2pt]
\extitem Upon receiving $(\sid, \textsc{URS})$ from $P\in\mbf{P}\setminus\mbf{P}_\msf{corr}$ or $\mc{S}$, it does:
\begin{enumerate}
\item If $\msf{urs}=\bot$, it samples $r\overset{\$}{\leftarrow}\{0,1\}^\lambda$ and sets $\msf{urs}\leftarrow r$.
\item It reads the time $\Cl$ from $\mc{G}_\msf{clock}$.
\item If $f^P_\msf{wait}=0$, it sets $f^P_\msf{wait}\leftarrow1$.
\item If $t_\msf{start}=\bot$, it sets $t_\msf{start}\leftarrow\Cl$ and sends $(\sid,\textsc{Start},P)$ to $\mc{S}$.
\item If the request was made by $P\in\mbf{P}\setminus\mbf{P}_\msf{corr}$ and $P$ has sent a $(\sid_C,\textsc{Advance\_Clock})$ message in round $t_\msf{start}+\Delta$ (hence, $\Cl\geq t_\msf{start}+\Delta$), it sends $(\sid, \textsc{URS},\msf{urs})$ to $P$.
\item If the request was made by $\mc{S}$ and $\Cl\geq t_\msf{start}+\Delta-\alpha$, it sends $(\sid, \textsc{URS},\msf{urs})$ to $\mc{S}$.
\end{enumerate}
\vspace{2pt}
\extitem Upon receiving $(\sid_C,\textsc{Advance\_Clock})$ from $P\in\mbf{P}\setminus\mbf{P}_\msf{corr}$, it does:
\begin{enumerate}
    \item It reads the time $\Cl$ from $\mc{G}_\msf{clock}$. If $\Cl=t_\msf{start}+\Delta$ and $f^P_\msf{wait}=1$, and this is the first time that $P$ has sent a $(\sid_C,\textsc{Advance\_Clock})$ message during round $\Cl$, it sends $(\sid, \textsc{URS},\msf{urs})$ to $P$.

\item It returns $(\sid_C, \textsc{Advance\_Clock})$ to $P$ with destination identity $\mc{G}_\msf{clock}$.
\end{enumerate}
\end{small}
\end{tcolorbox}

\captionof{figure}{The DURS functionality \(\mathcal{F}_\msf{DURS}\) interacting with the simulator \(\mathcal{S}\), parameterized by delay $\Delta$ and simulator advantage $\alpha$.}
\label{fig:urs}
\end{figure}

\begin{figure}[H]
\begin{tcolorbox}[enhanced, colback=white, arc=10pt, drop shadow southeast]
\noindent\emph{\underline{The DURS protocol $\Pi_\msf{DURS}(\mc{F}^{\Phi,\Delta-\Phi,\alpha}_\msf{SBC},\mc{F}_\msf{RBC},\mbf{P})$.}}\\[5pt]
\begin{small}
Each party $P$ maintains a variable $\msf{urs}^P$ initialized to $\bot$ and two flags $f^P_\msf{wait}$, $f^P_\msf{awake}$, initialized to $0$.\\[2pt]
\extitem Upon receiving $(\sid, \textsc{URS})$ from $\mc{Z}$, the party $P$ does:
\begin{enumerate}[leftmargin=*]
\item If $\msf{urs}^P\neq\bot$, it returns $(\sid,\textsc{URS},\msf{urs}^P)$ to $\mc{Z}$. Otherwise,
\begin{enumerate}
    \item If $f^P_\msf{wait}=0$, she sets $f^P_\msf{wait}\leftarrow 1$.
    \item If $f^P_\msf{awake}=0$, she sends $(\sid,\textsc{Broadcast},\texttt{Wake\_Up})$ to $\mc{F}^P_\msf{RBC}$.
\end{enumerate} 
\end{enumerate}
\vspace{2pt}
\extitem Upon receiving $(\sid, \textsc{Broadcast},\texttt{Wake\_Up},P^*)$ from $\mc{F}^{P^*}_\msf{RBC}$, if $P^*\in\mbf{P}$ and $f^P_\msf{awake}=0$, the party $P$ does:
\begin{enumerate}[leftmargin=*]
    \item\label{it:wake_up1} She sets $f^P_\msf{awake}\leftarrow 1$.
    \item\label{it:wake_up2} She chooses a randomness $\rho\overset{\$}{\leftarrow}\{0,1\}^\lambda$.
    \item\label{it:wake_up3} She sends $(\sid,\textsc{Broadcast},\rho)$ to $\mc{F}^{\Phi,\Delta-\Phi,\alpha}_\msf{SBC}$.
\end{enumerate}
\vspace{2pt}
\extitem Upon receiving $(\sid_C, \textsc{Advance\_Clock})$ from $\mc{Z}$, the party $P$ does:
\begin{enumerate}
    \item If $f^P_\msf{awake}=0$, she sends $(\sid_C, \textsc{Advance\_Clock})$ to $\mc{F}^P_\msf{RBC}$.
    \begin{enumerate}
        \item If $\mc{F}^P_\msf{RBC}$ responds with $(\sid, \textsc{Broadcast}, \texttt{Wake\_Up},P)$, she executes steps~\ref{it:wake_up1}-\ref{it:wake_up3} from the $\textsc{Broadcast}$ interface above.
        \item She sends $(\sid_C, \textsc{Advance\_Clock})$ to $\mc{G}_\msf{clock}$.
    \end{enumerate}
    Otherwise, she sends $(\sid_C, \textsc{Advance\_Clock})$ to $\mc{F}^{\Phi,\Delta-\Phi,\alpha}_\msf{SBC}$. Upon receiving the token back from $\mc{F}^{\Phi,\Delta-\Phi,\alpha}_\msf{SBC}$ (because the functionality returns $(\sid_C, \textsc{Advance\_Clock})$), she sends $(\sid_C, \textsc{Advance\_Clock})$ to $\mc{G}_\msf{clock}$.
\end{enumerate}
\vspace{2pt}
\extitem Upon receiving $(\sid,\textsc{Broadcast},\langle\rho_1,\ldots\rho_k\rangle)$  from $\mc{F}^{\Phi,\Delta-\Phi,\alpha}_\msf{SBC}$, if $\msf{urs}^P=\bot$, the party $P$ does:
\begin{enumerate}[leftmargin=*]
    \item She sets $\msf{urs}^P\leftarrow\bigoplus_{i\in[k]:\rho_i\in\{0,1\}^\lambda}\rho_i$.
    \item If $f^P_\msf{wait}=1$, she sends $(\sid,\textsc{URS},\msf{urs}^P)$ to $\mc{Z}$.
\end{enumerate}
\end{small}

\end{tcolorbox}

\captionof{figure}{The DURS protocol $\Pi_\msf{DURS}$ with parties in $\mbf{P}$.}
\label{fig:real_urs}
\end{figure}
\noindent
\textbf{The DURS protocol.} As a first application, we propose a protocol that employs SBC to realize the DURS functionality described above. The idea is simple: each party contributes its randomness by broadcasting it via SBC to other parties. After SBC is finalized (with delay $\Delta$), all parties agree on the XOR of the received random strings as the generated URS. In addition, the parties agree on the beginning of the URS generation period via a special `\texttt{Wake\_Up}' message broadcast in RBC manner by the first activated party.

Below, we prove that $\Pi_\msf{DURS}$ achieves UC-secure DURS generation (cf. proof in Appendix~\ref{app:DURS_proof}).

\begin{theorem}\label{thm:DURS}
 Let $\Delta,\Phi,\alpha$ be non-negative integers such that $\Delta>\Phi>0$ and $\Delta-\Phi\geq\alpha$. The protocol $\Pi_\msf{DURS}$ in Figure~\ref{fig:real_urs} UC-realizes $\mc{F}^{\Delta,\alpha}_\msf{DURS}$ in the $(\mc{F}^{\Phi,\Delta-\Phi,\alpha}_\msf{SBC},\mc{F}_\msf{RBC},\mc{G}_\msf{clock})$-hybrid model against an adaptive adversary corrupting $t<n$ parties.
\end{theorem}

\subsection{Self-tallying e-voting}\label{STE}
\begin{figure}[H]
  \centering

  \begin{tcolorbox}[enhanced, colback=white, arc=10pt, drop shadow southeast]
    \noindent\emph{\underline{The voting system functionality $\mathcal{F}^{\Phi, \Delta, \alpha}_{VS}(\mbf{V})$}}\\[5pt]
\begin{small}
    The functionality initializes the list $L_\msf{cast}$ of cast votes and four variables $t_\msf{cast}^\msf{start}$,  $t_\msf{cast}^\msf{end}$, $t_\msf{tally}$, and $res$ as $\bot$.
    It also maintains the set of corrupted voters, $\mbf{V}_\msf{corr}$, initialized as empty.
    
    \extitem Upon receiving $(\sid, \textsc{Init})$ from the last authority it does:
    \begin{enumerate}
    \item It reads the time $t_\msf{cast}^{\msf{start}}$ from $\mc{G}_\msf{clock}$.
    \item It sets the time $t_\msf{cast}^{\msf{end}} \leftarrow t_\msf{cast}^{\msf{start}} + \Phi$. 
    \item It sets the time $t_{\msf{tally}} \leftarrow t_\msf{cast}^{\msf{end}} + \Delta$. 
    \end{enumerate}
    \vspace{2pt}

    \extitem Upon receiving $(\sid, \textsc{Vote}, v)$ from voter $V \not\in \mbf{V}_\msf{corr}$ or $(\sid, \textsc{Vote}, v)$ from $\mc{S}$ on behalf of $V \in \mbf{V}_\msf{corr}$ it does:
    \begin{enumerate}
    \item It reads the time $\Cl$ from $\mc{G}_\msf{clock}$.
    \item If $t_\msf{cast}^{\msf{start}} \le \Cl < t_\msf{cast}^{\msf{end}}$, then
      \begin{enumerate}
      \item It picks a unique random $\msf{tag}$ from $\{0, 1\}^\lambda$.
      \item If $V \in \mbf{V} \setminus \mbf{V}_\msf{corr}$, it checks if $v$ is a valid vote and adds $(\msf{tag},v,V,\Cl,0)$ to $L_\msf{cast}$ and sends $(\sid, \textsc{Vote}, \msf{tag}, V)$ to $\mc{S}$. Otherwise, it it checks if $v$ is a valid vote and adds $(\msf{tag},v,V,\Cl,1)$ to $L_\msf{cast}$ and sends $(\sid, \textsc{Vote}, \msf{tag}, v, V)$ to $\mc{S}$.
      \end{enumerate}
    \end{enumerate}
    \vspace{2pt}

    \extitem Upon receiving $(\sid, \textsc{Corruption\_Reqest})$ from $\mc{S}$, it sends $(\sid, \textsc{Corruption\_Reqest}, \langle(\msf{tag}, v, V, \Cl^*, 0) \in L_\msf{cast} : V \in \mbf{V}_\msf{corr}\rangle)$ to $\mc{S}$.
    \vspace{2pt}

    \extitem Upon receiving $(\sid, \textsc{Allow}, \msf{tag}, \tilde v , \tilde V)$ from $\mc{S}$, it does:
    \begin{enumerate}
    \item It reads the time $\Cl$ from $\mc{G}_\msf{clock}$.
    \item If $t_\msf{cast}^\msf{start} \le \Cl < t_\msf{cast}^\msf{end}$ and there is a tuple $(\msf{tag}, v, \tilde V, \Cl^*, 0) \in L_\msf{cast}$ and $V \in \mbf{V}_\msf{corr}$, it it checks if $v$ is a valid vote, updates the tuple as $(\msf{tag}, \tilde v , \tilde V, \Cl^*, 1)$, and sends $(\sid, \textsc{Allow\_OK})$ to $\mc{S}$. Otherwise, it ignores the message.
    \end{enumerate}
    \vspace{2pt}

    \extitem Upon receiving $(\sid_C, \textsc{Advance\_Clock})$ from $V \not\in \mbf{V}_\msf{corr}$, it does:
    \begin{enumerate}
    \item It reads the time $\Cl$ from $\mc{G}_\msf{clock}$
    \item If $\Cl=t_\msf{tally}-\alpha$ and $res=\bot$, it does:
      \begin{enumerate}
      \item It updates every tuple $(\cdot, \cdot, V^*, \cdot, 0) \in L_\msf{cast}$ such that $V^* \in \mbf{V}\setminus \mbf{V}_\msf{corr}$ as $(\cdot, \cdot, V^*, \cdot, 1)$ (to guarantee that votes cast by honest voters are tallied).
      \item And it computes the tallying function on all the stored votes $v$ such $(\msf{tag}, v, \cdot, \cdot, 1)\in L_\msf{cast}$. If a voter $V$ has multiple such stored votes in excess of his quota, choose the most recent quota-sized subset of votes. Set $res \leftarrow f(v_1,\dots,v_n)$.
      \item It sends $(\sid, \textsc{Result}, res)$ to $\mc{S}$.
      \end{enumerate}
    \item If $\Cl=t_\msf{tally}$, if this is the first time it has received a $(\sid_C,\textsc{Advance\_Clock})$ message from $V$ during round $\Cl$, then it does:
      \begin{enumerate}
      \item It sends $(\sid, \textsc{Result}, res)$ to $V$.
      \item It returns $(\sid_C, \textsc{Advance\_Clock})$ to $P$ with destination identity $\mc{G}_\msf{clock}$.
      \end{enumerate}
    \end{enumerate}
    \end{small}
  \end{tcolorbox}
  \captionof{figure}{Preneel and Szepieniec's voting system functionality~\cite{SzepieniecP15} adapted to the global clock ($\mc{G}_{\msf{clock}}$) model and adaptive corruption.}
  \label{fig:F_VS}
\end{figure}

The concept of self-tallying elections was introduced by Kiayias and Yung in~\cite{KiayiasY02}. In this paradigm, the post-ballot-casting (tally) phase can be performed by any interested party, removing the need for tallier designation. It was further improved by Groth in~\cite{Groth04}, and later studied in the UC framework by Szepieniec and Preneel in~\cite{SzepieniecP15}. To ensure fairness, \emph{i.e.} to prevent intermediary results from being leaked before the end of the casting phase, all these previous works introduce a \emph{trusted control voter} that casts a dummy ballot last, contradicting self-tallying in some sense. In this section, we deploy our simultaneous broadcast channel to solve the fairness challenge in self-tallying elections, lifting the need for this trusted control voter.\medskip

\noindent\textbf{The voting system (VS) functionality.} The ideal voting system functionality, $\mc{F}^{\Phi, \Delta, \alpha}_{VS}$ is presented in Figure~\ref{fig:F_VS}. It is simply the adaptation of Preneel and Szepieniec's functionality to the global clock model and adaptive corruption~\cite{SzepieniecP15}. The VS functionality only differs from the SBC functionality in that the individually broadcast ballots are not forwarded to the voters (and simulator), but instead the result (tally) of the election is sent to them.

Voters submit their votes to the functionality during the casting period that lasts $\Phi$ amount of time from the opening of the election. The functionality does not allow the  adversary to read the honestly cast votes or to falsify them. The functionality releases the result of the election when it moves to the tally phase after $\Phi+\Delta$ time has elapsed from the opening of the election. The simulator has an advantage $\alpha$, i.e., it can obtain the election result (on behalf of some corrupted party) when $\Phi+\Delta-\alpha$ time has elapsed since the opening of the election, (yet after the end of the casting period).  The VS functionality is presented in detail in Figure~\ref{fig:F_VS}.

\noindent\textbf{The self-tallying VS (STVS) protocol.} We deploy an SBC instead of the BB used in the original protocol of Preneel and Szepieniec~\cite{SzepieniecP15} to ensure fairness, removing the need for the control dummy party. The protocol assumes a public-key generation mechanism formalised by an ideal functionality $\mc{F}_\msf{SKG}$ for the authorities public key and corresponding private key shares, and a voters' key generation functionality $\mc{F}_\msf{PKG}$ for eligibility.

\begin{figure}[H]
  \centering

  \begin{tcolorbox}[enhanced, colback=white, arc=10pt, drop shadow southeast]
    \noindent\emph{\underline{The self-tallying protocol $\Pi_{\mathsf{STVS}}(\mc{F}^{\Phi, \Delta, \alpha}_\msf{SBC}, \mc{F}_\msf{RBC}, \mathcal{F}_{\mathsf{PKG}}, \mathcal{F}_{\mathsf{SKG}}, \mbf{V})$}}\\
\begin{small}
    \extitem Initiate by invoking $\mathcal{F}_{\mathsf{PKG}}$ followed by $\mathcal{F}_{\mathsf{SKG}}$.
    \vspace{2pt}

    \extitem All authorities $A_j$ choose random values $x_{i, j} \leftarrow \mathbb{Z}_{n^2}$ for all voters $V_i$ such that $\sum_i x_{i, j} = 0$. They send these values to $\mathcal{F}_{\mathsf{RBC}}$ but encrypted with that voter's public key. Also,they publish $w^{x_{i,j}}$ for each $x_{i,j}$.
    \vspace{2pt}

    \extitem The scrutineers check that $\sum_i x_{i, j} = 0$ for all $j$ by calculating $\prod_i w^{x_{i, j}} \stackrel{?}{=} 1$. Also, the scrutineers calculate every voter's verification key $w_i = w^{\sum_j x_{i,j}} = \prod_j w^{x_{i,j}}$.
    \vspace{2pt}

    \extitem All voters $V_i$ read their messages from the authorities and determine their own secret exponent $x_i = \sum_j x_{i, j}$.
    \vspace{2pt}
    
    \extitem In order to vote, the voters select a public random seed $r$, for example by querying the random oracle. Next, each voter encrypts his vote using $r^{x_i}$ for randomizer. This encryption is posted to \st{$\mc{F}_\msf{BB}$} $\mc{F}^{\Phi, \Delta, \alpha}_\msf{SBC}$ along with a proof that the ballot encrypts an allowable vote and that the correct secret exponent was used, and with a signature on the previous two objects. \st{When all real voters have cast their vote, the control voter casts his dummy vote and proves that it does not alter the tally and that it uses the correct secret exponent.}
    \vspace{2pt}

    \extitem \st{All participants send a $(\sid , \textsc{Read})$ message to $\mc{F}_\msf{BB}$} Upon receiving $(\sid, \textsc{Broadcast}, \langle b_1, \ldots, b_k\rangle)$ from $\mc{F}^{\Phi, \Delta, \alpha}_\msf{SBC}$, voters combine all the votes and calculate the tally $res$.
\end{small}
  \end{tcolorbox}
  \caption{The self-tallying protocol $\Pi_\msf{STVS}$ with voters $\mbf{V}$. Variant of Preneel and Szepieniec's self-tallying protocol~\cite{SzepieniecP15}, where instead of posting ballots to the bulletin board $\mc{F}_\msf{BB}$, they are posted via SBC $\mc{F}_\msf{SBC}$, removing the need for the trusted control voter.}
  \label{fig:PI_STVS}
\end{figure}

The proof of UC-security of $\Pi_\msf{STVS}$ is similar to the original one~\cite{SzepieniecP15}.

\begin{theorem}\label{thm:STVS}
  Let $\Delta,\Phi,\alpha$ be non-negative integers such that $\Delta>\Phi>0$ and $\Delta\geq\alpha$. The protocol $\Pi_\msf{STVS}$ in Figure~\ref{fig:PI_STVS} UC-realizes $\mc{F}^{\Phi, \Delta,\alpha}_\msf{VS}$ in the $(\mc{F}^{\Phi,\Delta,\alpha}_\msf{SBC},\mc{F}_\msf{RBC},\mathcal{F}_{\mathsf{PKG}}, \mathcal{F}_{\mathsf{SKG}},\mc{G}_\msf{clock})$-hybrid model against an adaptive adversary corrupting $t<n$ parties.
\end{theorem}

%% file: appendix.tex
\input{appendix_UBC.tex}

\input{appendix_FBC_proof.tex}

\input{appendix_TLE_proof.tex}

\input{appendix_DURS_proof.tex}

%% file: appendix_UBC.tex
\section{Proof of Lemma~\ref{lem:real_UBC}}\label{app:UBC}
\begin{proof}
By Fact~\ref{fact:RBC}, $\mc{F}_\msf{RBC}$ can be realized in the $(\mc{F}_\msf{cert},\mc{G}_\msf{clock})$-hybrid model. By the Universal Composition Theorem, it suffices to show that $\Pi_\msf{UBC}$ described in Figure~\ref{fig:real_UBC} UC-realizes $\mc{F}_\msf{UBC}$ in the $(\mc{F}_\msf{RBC},\mc{G}_\msf{clock})$-hybrid model.

We devise a simulator $\mc{S}_\msf{UBC}$ that given a real-world adversary $\mc{A}$, simulates an execution of $\Pi_\msf{UBC}$ in the $(\mc{F}_\msf{RBC},\mc{G}_\msf{clock})$-hybrid world in the presence of $\mc{A}$. $\mc{S}_\msf{UBC}$ acts as a proxy forwarding messages from/to the environment $\mc{Z}$ to/from $\mc{A}$. It also simulates the behavior of honest parties, and controls the instances of $\mc{F}_\msf{RBC}$.

In particular, $\mc{S}_\msf{UBC}$, for every honest party $P$, maintains the counters $\msf{total}^P,\msf{count}^P$ as described in the protocol $\Pi_\msf{UBC}$. It operates as follows.\\[2pt]
\extitem Upon receiving $(\sid,\textsc{Broadcast},\msf{tag},M,P)$ from $\mc{F}_\msf{UBC}$, it initiates the broadcasting of $M$ on behalf of (honest party) $P$ via $\mc{F}^{P,\msf{total}^P}_\msf{RBC}$. In particular, it provides $\mc{A}$ with $(\sid,\textsc{Broadcast},M,P)$ as leakage of $\mc{F}^{P,\msf{total}^P}_\msf{RBC}$ and records the tuple $(\msf{tag},M,P,\mc{F}^{P,\msf{total}^P}_\msf{RBC})$.\\[2pt]
\extitem Upon receiving as $\mc{F}^{P,i}_\msf{RBC}$ a message $(\sid,\textsc{Broadcast},M,P)$ from $\mc{A}$, it does:
\begin{enumerate}
 \item If $\mc{F}^{P,i}_\msf{RBC}$ is terminated or $P$ is honest or there is a recorded tuple $(\msf{tag},M_i,P,\mc{F}^{P,i}_\msf{RBC})$, it ignores the message.
 \item It terminates $\mc{F}^{P,i}_\msf{RBC}$ and sends $(\sid,\textsc{Broadcast},M,P)$ to all parties and $\mc{A}$. 
 \item It sends $(\sid,\textsc{Broadcast},M,P)$ to $\mc{F}_\msf{UBC}$.
 \end{enumerate}
\extitem Upon receiving as $\mc{F}^{P,i}_\msf{RBC}$ a message $(\sid,\textsc{Allow},\tilde{M})$ from $\mc{A}$, it does:
\begin{enumerate}
 \item If $\mc{F}^{P,i}_\msf{RBC}$ is terminated or $P$ is honest, it ignores the message.
 \item It finds the tuple $(\msf{tag},M_i,P,\mc{F}^{P,i}_\msf{RBC})$ that corresponds to $\mc{F}^{P,i}_\msf{RBC}$ (i.e., the $i$-th honest broadcast of $P$). Note that such tuple exists, as $\mc{F}^{P,i}_\msf{RBC}$ is not terminated.
 \item It terminates $\mc{F}^{P,i}_\msf{RBC}$ and sends $(\sid,\textsc{Broadcast},\tilde{M},P)$ to all parties and $\mc{A}$. 
 \item It sends $(\sid,\textsc{Allow},\msf{tag},\tilde{M})$ to $\mc{F}_\msf{UBC}$.
\end{enumerate}
\vspace{2pt}
\extitem Upon receiving $(\sid_C,\textsc{Advance\_Clock},P)$ from $\mc{G}_\msf{clock}$, it reads the time $\Cl$ from $\mc{G}_\msf{clock}$ and completes the simulation of the broadcasting of $P$'s messages during round $\Cl$. Namely, for $j=1,\ldots,\msf{count}^P$, by playing the role of $\mc{F}^{P,\msf{total}^P-(\msf{count}^P-j)}_\msf{RBC}$, it does: 
\begin{enumerate}
\item It finds the corresponding tuple $(\msf{tag},M_j,P,\mc{F}^{P,\msf{total}^P-(\msf{count}^P-j)}_\msf{RBC})$.
\item It sends $(\sid,\textsc{Broadcast},M_j,P)$ to all parties and $\mc{A}$.
\end{enumerate} 
Then, it resets $\msf{count}^P$ to $0$.\\[2pt]
\vspace{5pt}
By the description of $\mc{S}_\msf{UBC}$, we have that the simulation is perfect. This completes the proof.
\end{proof}

%% file: appendix_FBC_proof.tex
\section{Proof of Lemma~\ref{lem:real_FBC}}\label{app:FBC_proof}

\begin{proof}
We devise a simulator $\mc{S}_\msf{FBC}$ that given a real-world adversary $\mc{A}$, simulates an execution of $\Pi_\msf{FBC}$ in the $(\mc{F}_\msf{UBC},\mc{W}_q(\mc{F}^*_\msf{RO}),\mc{F}_\msf{RO},\mc{G}_\msf{clock})$-hybrid world in the presence of $\mc{A}$. $\mc{S}_\msf{FBC}$ acts as a proxy forwarding messages between the environment $\mc{Z}$ and $\mc{A}$. It also simulates the behaviour of honest parties, and controls the local functionalities $\mc{F}_\msf{UBC}$, $\mc{W}_q(\mc{F}^*_\msf{RO})$, and $\mc{F}_\msf{RO}$. 
In particular, $\mc{S}_\msf{FBC}$ operates as follows.\\[2pt]
\noindent\extitem For every party $P$ and every time $\Cl$, it initialises an empty list of tags $L_\msf{tag}$, representing the messages that the party $P$ sends $\mc{F}_\msf{FBC}$ to be broadcast during $\Cl$ while being honest.\\[2pt]
\noindent\extitem Every tuple $(\sid,\textsc{Broadcast},M)$ that $\mc{Z}$ sends to honest dummy party $P$, the simulator forwards to $\mc{F}_\msf{FBC}$.
It receives tuple $(\sid,\textsc{Broadcast},\msf{tag},P)$ back from $\mc{F}_\msf{FBC}$, and it adds $\msf{tag}$ to $L_\msf{tag}$. 
\\[2pt]
\noindent\extitem Upon receiving $(\sid_C,\textsc{Advance\_Clock},P)$ from $\mc{G}_\msf{clock}$, if there is at least one corrupted party, it does:
\begin{enumerate}
    \item While $P$ remains honest and $L_\msf{tag}$ is non empty, it simulates the first message in $L_\msf{tag}$ being broadcast by $P$ as follows:
    \begin{enumerate}
    \item It chooses a random value $\rho$ from the domain of $\mc{F}_\msf{RO}$.
    \item \label{HonestTLE}  It produces a TLE ciphertext $c$ of $\rho$ (this includes programming $\mc{F}^*_\msf{RO}$ on $q$ puzzle generation queries ($r_0,\ldots,r_{q-1}$));
    \item It chooses a random value $y$ from the image of $\mc{F}_\msf{RO}$;           
    \item  \label{permissionRequest} If $\mc{A}$ has already made a RO query from $\{r_0,\ldots,r_{q-1},\rho\}$, the simulation aborts. If not, the simulator sends $(\sid,\textsc{Output\_Request},\msf{tag})$ to $\mc{F}_\msf{FBC}$ which responds as
    $(\sid,\textsc{Output\_Request},\msf{tag},M,P,\Cl^{*})$, where $(\msf{tag},M,P,\Cl^{*})$ is the tuple corresponding to the ``locked'' message $M$. 
    
    \item \label{ProgrammingROhonest} It programs $\mc{F}_\msf{RO}$ on $\rho$ as $(\rho,M\oplus y)$, unless $\mc{A}$ has already queried $\rho$, where simulation aborts.
    
    \item It sends $(\sid,\textsc{Broadcast},(c,y))$ to $\mc{F}_\msf{FBC}$ and receives $(\sid,\textsc{Broadcast},\msf{tag},(c,y),P)$.

    \item It sends $(\sid,\textsc{Broadcast},\msf{tag},(c,y),P)$ to $\mc{A}$ as leakage of $\mc{F}_\msf{FBC}$. Note this is identical to the leakage via $\mc{F}_\msf{UBC}$ in $\Pi_\msf{FBC}$.
    
    \item  Upon receiving  permission message $(\sid,\textsc{Allow},(\tilde{\msf{tag}}, (\tilde{c},\tilde{y}),\tilde{P})$, if $\mc{A}$ has corrupted $P$, then it proceeds as follows:
    \begin{enumerate}
    
    \item If there is no tuple containing $\tilde{\msf{tag}}$ and $\Tilde{P}$ in the corresponding list $L_\msf{tag}$, then the simulator ignores this message.

    \item\label{item:sfbc_corrupted_invalid} If $(\tilde{c},\tilde{y})$ is not in the correct format (i.e., $\tilde{c}$ should be a TLE ciphertext with time difficulty $2$ and $\tilde{y}$ should be a value from the image of $\mc{F}_\msf{RO}$), then it simulates the broadcasting of $(\tilde{c},\tilde{y})$ via $\mc{F}_\msf{UBC}$ by sending $(\sid,\textsc{Broadcast},(\tilde{c},\tilde{y}))$ to all simulated parties and $(\sid,\textsc{Broadcast},(\tilde{c},\tilde{y}),P)$ to $\mc{A}$. It does not allow $\mc{F}_\msf{FBC}$ to do any broadcast at this point, as in the real-world, honest parties would ignore invalid messages. 
    
    \item If $(\tilde{c},\tilde{y})$ is in the correct format, then it obtains the decryption witness $\tilde{w}$ by solving the puzzle associated with $\tilde{c}$ and extracting the  queries $\tilde{r}_0,\ldots,\tilde{r}_{2q-1}$. If there is a query $\tilde{r}_w$ that has not been already made by $\mc{A}$, then it again broadcasts the message via $\mc{F}_\msf{UBC}$ as in Step~\ref{item:sfbc_corrupted_invalid} (in the real-world, the honest parties would ignore invalid decryptions caused by inconsistent puzzles). Note that as the simulator controls the random oracle functionality, it is able to make as many queries as it likes without imposing any delay. Hence it can make $2q$ queries without advancing the global clock.
    
    \item Using $\tilde{w}$, it decrypts $\tilde{c}$ as $\tilde{\rho}$.
    
    \item If $\mc{A}$ has never queried $\mc{F}_\msf{RO}$ for $\tilde{\rho}$, it programs $\mc{F}_\msf{RO}$ as ($\tilde{\rho},\tilde{M}\oplus \tilde{y}$), so the message can be extracted from  $\tilde{\rho}\oplus \Tilde{y}$.            Note that in this case, the adversary does not know the message $M$ - $\mc{A}$ has just chosen a random $\tilde{y}$. However, this does not make the message invalid. 

    \item \label{generateTLEciphertext} If $\mc{F}_\msf{RO}$ has already been programmed on $\tilde{\rho}$ as $(\tilde{\rho},\tilde{\eta})$, it obtains the message $\tilde{M}\leftarrow\tilde{y}\oplus\tilde{\eta}$.
    
    \item It simulates the broadcast of $(\tilde{c},\tilde{y})$ via $\mc{F}_\msf{UBC}$ by sending $(\sid,\textsc{Broadcast},(\tilde{c},\tilde{y}))$ to all simulated parties and $(\sid,\textsc{Broadcast},(\tilde{c},\tilde{y}),P)$ to $\mc{A}$.
    
    \item\label{item:sfbc_corrupted_OK} It sends $(\sid,\textsc{Allow},(\msf{tag},\tilde{M},P))$ to $\mc{F}_\msf{FBC}$ and receives confirmation response $(\sid,\textsc{Allow\_OK})$.
    
    \end{enumerate}

\item Otherwise, if $P$ remains honest, it simulates the broadcast of the message via $\mc{F}_\msf{UBC}$, sending $(\sid,\textsc{Broadcast},(c,y))$ to all simulated parties, and $(\sid,\textsc{Broadcast},(c,y), P)$ to $\mc{A}$.

\item It deletes this tag from $L_\msf{tag}$.
\end{enumerate}

\item It sends $(\sid_C,\textsc{Advance\_Clock})$ to $\mc{F}_\msf{FBC}$ on behalf of $P$, and receives a set of tuples of the form $(\sid,\textsc{Broadcast},M^*)$ back from $\mc{F}_\msf{FBC}$.

\item It forwards each tuple $(\sid,\textsc{Broadcast},M^*)$ it receives from $\mc{F}_\msf{FBC}$ to the environment on behalf of $P$. 
\end{enumerate}
\vspace{2pt}
\noindent\extitem Whenever $\mc{A}$ corrupts a party $\tilde{P}$, the simulator corrupts the corresponding dummy party. It also sends $(\sid,\textsc{Corruption\_Request})$ to $\mc{F}_\msf{FBC}$ and receives the response
\newline \noindent $(\sid,\textsc{Corruption\_Request},\langle(\msf{tag},M,P,\Cl^*)\in L_\msf{pend}: P\in\mbf{P}_\msf{corr}\rangle)$.

\noindent\extitem Upon receiving from $\mc{F}_\msf{UBC}$ a message $(\sid,\textsc{Broadcast},(c,y),P)$ from $\mc{A}$ on behalf of $P\in\mbf{P}_\msf{corr}$, it acts as in the Steps~\ref{item:sfbc_corrupted_invalid}-\ref{item:sfbc_corrupted_OK}.
\vspace{5pt}

We now analyse the security of the simulation. In the adaptive corruption model, we see from $\Pi_\msf{FBC}$ that in the real world, $\mc{A}$ can corrupt parties at any point before the advance clock command, or in between each message being broadcast by a party $P$ (i.e. in between steps $(4)(e)$ and $(5)$).

This is represented in the ideal world by the simulator checking after each message broadcast if the party is still honest.

\textbf{Equivocation}
In step \ref{HonestTLE}, the random oracle queries used to generate the TLE ciphertext are indistinguishable from random, as they are based on a random value $\rho$. This is the same as in the real world.
In step \ref{ProgrammingROhonest}, the random oracle is programmed with $\mc{F}_\msf{RO}(\rho,M\oplus y)$, where $M\oplus y$ is random as $y$ is selected randomly from the image of $\mc{F}_\msf{RO}$.

For a corrupted party, the simulator uses the corrupted message $\tilde{M}$ chosen by the adversary in order to equivocate with the honest message $M$, and hence emulate an instantiation of $\mc{F}_\msf{RO}$. This is seen in \ref{generateTLEciphertext}.
 
\textbf{Simulation aborts}
The simulation can abort if either the random oracle is queried for the same value twice, or a tag is selected twice.
As described in \ref{permissionRequest} and \ref{ProgrammingROhonest}, $\mc{S}$ will abort if $\mc{A}$ correctly guesses a random oracle query. 
We note that $\mc{A}$ is bounded by $q$ queries per round, and hence has $2q$ queries total.
The domain of the random oracle is of exponential size, and therefore the probability of $\mc{A}$ guessing the programming of the random oracle, and hence the simulator aborting is negligible.
Finally, for $\mc{A}$ to query the the correct value of $\rho$ without guessing at random would break the security of $\mc{F}_\msf{RO}$. 
Therefore the simulator aborts with negligible probability.

Each tag is selected at random by the fair broadcast functionality, from a domain of exponential size, and hence won't be re-sampled by $\mc{F}_\msf{FBC}$ with more than negligible probability. 

Therefore we have that the simulator emulates the protocol with all but negligible probability, completing the proof.
\end{proof}

%% file: appendix_TLE_proof.tex
\section{Proof of Theorem~\ref{thm:real_TLE}}\label{app:TLE_proof}
\begin{proof}
Below we present the simulator $\mathcal{S}$ that given the real world adversary $\mathcal{A}$ simulates the execution of the protocol $\Pi_{\msf{TLE}}$ in the $(\mc{W}_q(\mc{F}^*_\msf{RO}),\mc{F}_\msf{RO},\mc{F}^{\Delta,\alpha}_\msf{FBC}, \mc{G}_\msf{clock})$-hybrid model in the adaptive corruption setting.\\[2pt]
%
\extitem Upon receiving the initial corruption set $\mathbf{P}_{\msf{corr}}$ from $\mc{Z}$, $\mc{S}$ forwards it to $\mc{A}$ as if it was $\mc{Z}$. Upon receiving $\mathbf{P}_{\msf{corr}}$ from $\mc{A}$ as if it was $\mc{W}_q(\mc{F}^*_\msf{RO})$,$\mc{F}_\msf{RO}$ and $\mc{F}^{\Delta,\alpha}_\msf{FBC}$, $\mc{S}$ forwards it to $\Pi_{\msf{TLE}}$. In addition, $\mc{A}$ can adaptively corrupt parties and requested their internal state which means the messages that they have generate. In that case, $\mc{S}$ adaptively corrupts the same parties in the ideal world and issues a $\textsc{Leakage}$ command to $\mc{F}_{\msf{TLE}}$ for retrieving the messages of the corrupted parties (previously honest) and send them back to $\mc{A}$.
Whatever message $\mc{S}$ receives from $\mc{F}_{\msf{TLE}}$ on behalf of a corrupted party $P$, $\mc{S}$ records and forwards that message to $\mc{A}$ as if it was $P$ and returns after recording whatever receives from $\mc{A}$ to $\mc{F}_{\msf{TLE}}$. Moreover, $\mc{S}$ registers the encryption/decryption algorithms $(e_{\mc{S}},d_{\mc{S}})$, which are the same as in protocol $\Pi_{\msf{TLE}}$. Similar to the proof in~\cite{ALZ21}, the \emph{Extended encryption} is not the same, specifically the created ciphertexts $c_{2},c_{3}$ are equal to a random value. Observe that still the distribution of both $(c_{2},c_{3})$ in both executions is still the same as both $c_{2},c_{3}$ in the real protocol are random.\\[2pt]
\extitem Upon receiving $(\sid,\textsc{Enc},\tau,\msf{tag},\msf{Cl},0^{|M|},P)$  from $\mc{F}_{\msf{TLE}}$, $\mc{S}$ records $(\tau_{\msf{dec}},\msf{tag},\msf{Cl},0^{M},c,\textsf{uncorr},$ $\textsf{nobroadcast},P)$, where $c$ is the resulting ciphertext by using the algorithm $e_{\mc{S}}$. Observe that the ciphertext does not contain any information about the actual message.
He updates his list, named $L^{\mc{S}}_{\msf{RO}^{*}}$ (initially empty), for the generation of that ciphertext. Moreover, he updates his list for the second and the third argument of the encryption as if it was $\mc{F}_{\msf{RO}}$ (e.g $c_{2}$ and $c_{3}$). Then, he returns the token back to $\mc{F}^{\msf{leak,delay}}_{\msf{TLE}}$. In addition, $\mc{S}$ keeps a record of all the randomness that he used for each encryption.\\[2pt]
\extitem Upon receiving$(\sid,\textsc{Advance\_Clock},P)$ from $\mc{G}_{\msf{clock}}$ from an honest party $P$, $\mc{S}$ reads the time $\msf{Cl}$ from $\mc{G}_{\msf{clock}}$ and does:
\begin{enumerate}
    \item For every stored tuple $(\tau_{\msf{dec}_j},\msf{tag}_{j},\msf{Cl}_{j},0^{|M_{j}|},c_{j},\textsf{uncorr},\textsf{broadcast},\cdot)$, $\mc{S}$ updates its list, named $L^{\mc{S}}_{\msf{RO}^{*}}$ , with $q$ evaluation queries for solving the ciphertexts issued by honest parties on previous rounds,  as if it was $\mc{F}^{*}_{\msf{RO}}$ in the real protocol. $\mc{S}$ picks a random $\msf{tag}^{*}_{j}$ from $\{0,1\}^{\lambda}$, records the tuple $(\msf{tag}^{*}_{j},c_{j},P,\msf{Cl})$ in $L^{\mc{S}}_{\msf{pend}}$ and sends $(\sid,\textsf{Broadcast},\msf{tag}^{*}_{j},P)$ to $\mc{A}$ as if it was $\mc{F}_{\msf{FBC}}$.
    \item For every tuple $(\msf{tag},c,P^*,\msf{Cl}^{*})$  in $L^{\mc{S}}_{\msf{lock}}\cup L^{\mc{S}}_{\msf{pend}}$ such that $\msf{Cl}-\msf{Cl}^{*}=\Delta+1$, $\mc{S}$ removes it from $(\msf{tag},c,P^*,\msf{Cl}^{*})$ and updates the tuple $(\tau_{\msf{dec}},\msf{tag}^*,\msf{Cl}^*,0^{M},c,\cdot,\textsf{unbroadcast},P^*)$ to $(\tau_{\msf{dec}},\msf{tag}^*,\msf{Cl}^*,0^{M},c,\cdot,\textsf{broadcast},P^*)$. Then $\mc{S}$ sends $(\sid,\textsc{Update},\{(c_{j},\msf{tag}_{j})\}_{j=1}^{p(\lambda)})$ to $\mc{F}_{\msf{TLE}}$ for all tuples $(\tau_{\msf{dec}_j},\msf{tag}_j,\msf{Cl}^*,0^{M},c_j,\cdot,\textsf{broadcast},\cdot)$.
\end{enumerate}
\extitem Upon receiving $(\sid,\textsc{Output\_Request},\msf{tag})$ from $\mc{A}$ as if it was $\mc{F}_{\msf{FBC}}$, $\mc{S}$ reads the time $\msf{Cl}$ from $\mc{G}_{\msf{clock}}$ and searches his list $L^{\mc{S}}_{\msf{pend}}$ for a tuple of the form $(\msf{tag},c,P,\msf{Cl}^{*})$ such that $\msf{Cl}-\msf{Cl}^{*}=\Delta -\alpha$. Then, he removes $(\msf{tag},c,P,\msf{Cl}^{*})$ from $L^{\mc{S}}_{\msf{pend}}$ and adds it to $L^{\mc{S}}_{\msf{lock}}$. $\mc{S}$ returns $(\sid,\textsc{Output\_Request},\msf{tag},c)$ to $\mc{A}$ as if it was $\mc{F}_{\msf{FBC}}$.\\[2pt]
\extitem Upon receiving $(\sid,\textsc{Corruption\_Request})$ from $\mc{A}$ as if it was $\mc{F}_{\msf{FBC}}$, the simulator sends $(\sid,\textsc{Corruption\_Request},\langle(\msf{tag},M,P,\Cl^*)\in L_\msf{pend}: P\in\mbf{P}_\msf{corr}\rangle)$ to $\mc{A}$ as if it was $\mc{F}_{\msf{FBC}}$. Observe that $\mc{S}$ has access to all of these tuples as he records all messages between $\mc{A}$ and the corrupted parties.\\[2pt]
\extitem Upon receiving $(\sid,\textsc{Allow},\msf{tag},\tilde{M},\tilde{P})$ from $\mc{S}$, it does:
\begin{enumerate}
    \item If there is no tuple  $(\msf{tag},M,\tilde{P},\Cl^*)$ in $L_\msf{pend}$ or $L_\msf{lock}$, it ignores the message.
        \item If $\tilde{P}\in\mbf{P}\setminus\mbf{P}_\msf{corr}$ or $(\msf{tag},M,\tilde{P},\Cl^*)\in L_\msf{lock}$, it ignores the message.
    \item If $\tilde{P}\in\mbf{P}_\msf{corr}$ and $(\msf{tag},M,\tilde{P},\Cl^*)\in L_\msf{pend}$ (i.e., the message is not  locked), it sets $\msf{Output}\leftarrow \tilde{M}$. If there is no tuple $(\msf{tag},\cdot,\cdot,\cdot)$ in $L_\msf{lock}$, it adds $(\msf{tag},\msf{Output},\tilde{P},\Cl^*)$ to $L_\msf{lock}$ and removes $(\msf{tag},M,\tilde{P},\Cl^*)$ from $L_\msf{pend}$. It sends $(\sid,\textsc{Allow\_OK})$ to $\mc{S}$.
\end{enumerate}

The $\mc{S}$ handles the Random oracle requests similarly as in~\cite[the proof of Theorem 1]{ALZ21}. Specifically, each time $\mc{S}$ receives a random oracle request, he checks if the requested value is the plaintext of any of the recorded $c_1$ ciphertexts. If yes, then $\mc{S}$ sends $(\sid,\textsc{Leakage})$ to $\mc{F}_{\msf{TLE}}$. After $\mc{S}$ learns the actual message he equivocates his response such that the ciphertext $c_2$ opens to the correct message. There is a probability that the requested value is the decryption of a $c_1$ ciphertext but the related message is not yet available to $\mc{S}$. In that case, $\mc{S}$ does not know how he should equivocate and the simulation fails. Observe that this scenario happens with negligible probability as it is only possible if $\mc{A}$ guesses the plaintext value of $c_1$ before solving the puzzle. Even if $\mc{A}$ solves the puzzle for the ciphertext $c_1$ faster than the honest parties due to the advantage he posses on receiving a message from $\mc{F}_{\msf{FBC}}$, the same advantage shares the simulator $\mc{S}$ in the ideal world.
Observe that, the distribution of messages in both worlds is indistinguishable and that completes the proof. 

Regarding broadcasting of malicious ciphertexts, the simulator acts as follows. Acting as $\mc{F}_{\msf{FBC}}$, upon receiving $(\sid, \textsc{Broadcast}, (c, \tau), P)$ from $\mc{A}$ on behalf of a corrupted party $P$, the simulator applies TLE decryption algorithm (including the generation of the time-lock puzzle solution) and computes the plaintext $M$. It sends the command $(\sid, \textsc{Update}, \{(c, M, \tau)\})$ to $\mc{F}_{\msf{TLE}}$.
\end{proof}

%% file: appendix_DURS_proof.tex
\section{Proof of Theorem~\ref{thm:DURS}}\label{app:DURS_proof}

\begin{proof}
We devise a simulator $\mc{S}_\msf{DURS}$ that given a real-world adversary $\mc{A}$, simulates an execution of $\Pi_\msf{DURS}$ in the $(\mc{F}^{\Phi,\Delta-\Phi,\alpha}_\msf{SBC},\mc{F}_\msf{RBC},\mc{G}_\msf{clock})$-hybrid world in the presence of $\mc{A}$. $\mc{S}_\msf{DURS}$ acts as a proxy forwarding messages from/to the environment $\mc{Z}$ to/from $\mc{A}$. It also simulates the behavior of honest parties, and controls the local functionalities $\mc{F}^{\Phi,\Delta-\Phi,\alpha}_\msf{SBC}$ and (all the instances of) $\mc{F}_\msf{RBC}$.

In particular, $\mc{S}_\msf{SBC}$ initializes a flag $f^P_\msf{awake}$ to $0$ for every party $P\in\mbf{P}$ and operates as follows.\\[2pt]
\extitem Upon receiving $(\sid,\textsc{Start},P)$ from $\mc{F}^{\Delta,\alpha}_\msf{DURS}$, where $P$ is honest, it simulates broadcast of $\texttt{Wake\_Up}$ via $\mc{F}^P_\msf{RBC}$ by providing $(\sid,\textsc{Broadcast},\texttt{Wake\_Up},P)$ to $\mc{A}$ as leakage from $\mc{F}^P_\msf{RBC}$. Then,
    \begin{enumerate}
        \item If it receives a message $(\sid,\textsc{Advance\_Clock},P)$ from $\mc{G}_\msf{clock}$ and $P$ remains honest and $f^P_\msf{awake}=0$, it simulates broadcast of   $\texttt{Wake\_Up}$ via $\mc{F}^P_\msf{RBC}$ by sending $(\sid,\textsc{Broadcast},\texttt{Wake\_Up},P)$ to all parties and $\mc{A}$ and terminates $\mc{F}^P_\msf{RBC}$.
        \item Upon receiving as $\mc{F}^P_\msf{RBC}$ a message $(\sid, \textsc{Allow},\tilde{M})$ from $\mc{A}$, if, meanwhile, $P$ has been corrupted, it simulates broadcast of $\tilde{M}$ via $\mc{F}^P_\msf{RBC}$ by sending $(\sid,\textsc{Broadcast},\tilde{M},P)$ to all parties and $\mc{A}$ and terminates $\mc{F}^P_\msf{RBC}$.
    \end{enumerate}
\vspace{2pt}
\extitem Upon receiving as $\mc{F}^{\tilde{P}}_\msf{RBC}$ a message $(\sid,\textsc{Broadcast},\texttt{Wake\_Up},\tilde{P})$ from $\mc{A}$, if $\tilde{P}$ is corrupted and $\mc{F}^{\tilde{P}}_\msf{RBC}$ has not fixed any sender party yet, it simulates broadcast of $\texttt{Wake\_Up}$ via $\mc{F}^{\tilde{P}}_\msf{RBC}$ by sending $(\sid,\textsc{Broadcast},\texttt{Wake\_Up},\tilde{P})$ to all parties and $\mc{A}$ and terminates $\mc{F}^{\tilde{P}}_\msf{RBC}$.\\[2pt]
\extitem Upon receiving as any honest party $P$ a message $(\sid,\textsc{Broadcast},\texttt{Wake\_Up},P^*)$ from $\mc{F}^{P^*}_\msf{RBC}$ (recall that by the description of $\mc{F}_\msf{RBC}$, the adversary cannot change the sender that invoked $\mc{F}^{P^*}_\msf{RBC}$), if $f^P_\msf{awake}=0$, it does:
\begin{enumerate}
    \item It sets $1\leftarrow f^P_\msf{awake}$.
    \item It simulates a broadcast of $P$'s randomness by providing $(\sid,\textsc{Sender}, \msf{tag}, 0^\lambda,P)$ to $\mc{A}$ as leakage from $\mc{F}^{\Phi,\Delta-\Phi,\alpha}_\msf{SBC}$, where $\msf{tag}$ is some unique random value in $\{0,1\}^\lambda$. Note that by SBC security, $\mc{S}_\msf{DURS}$ does not have to ``commit'' to $P$'s randomness at this point.  
    \item  Upon receiving as $\mc{F}^{\Phi,\Delta-\Phi,\alpha}_\msf{SBC}$ a message $(\sid,\textsc{Allow},\msf{tag},\tilde{\rho},P)$ from $\mc{A}$ during the SBC period $[t_\msf{start},t_\msf{end}]$ (specified as in Figure~\ref{fig:SBC}), if, meanwhile, $P$ has been corrupted and $\mc{S}_\msf{DURS}$ has never received an $\textsc{Allow}$ message for $\msf{tag}$, it records $\tilde{M}$ and sends $(\sid,\textsc{Allow\_OK})$ to $\mc{A}$.
\end{enumerate}
\vspace{2pt}
\extitem Upon receiving as $\mc{F}^{\Phi,\Delta-\Phi,\alpha}_\msf{SBC}$ a message $(\sid,\textsc{Broadcast},\tilde{\rho},\tilde{P})$ from $\mc{A}$ during the SBC period $[t_\msf{start},t_\msf{end}]$ (specified as in Figure~\ref{fig:SBC}), it records $\tilde{\rho}$ and returns $(\sid,\textsc{Sender},\msf{tag},\tilde{\rho},\tilde{P})$ to $\mc{A}$, where $\msf{tag}$ is some unique random value in $\{0,1\}^\lambda$.\\[2pt]
\extitem Upon receiving $(\sid,\textsc{Advance\_Clock},P)$ from $\mc{G}_\msf{clock}$, it reads the time $\Cl$ from $\mc{G}_\msf{clock}$. If $\Delta-\alpha$ time has elapsed since the time that it received $(\sid,\textsc{Start},P)$ from $\mc{F}^{\Delta,\alpha}_\msf{DURS}$ (i.e., $\Phi+(\Delta-\Phi)-\alpha$ time has elapsed since the beginning of the SBC simulations), $\mc{S}_\msf{DURS}$ commits to a consistent vector of parties' randomness contributions as follows:
\begin{enumerate}
    \item It sends $(\sid,\textsc{URS})$ to $\mc{F}^{\Delta,\alpha}_\msf{DURS}$ and receives the response $(\sid,\textsc{URS},r)$.
    \item Let $\tilde{\rho}_1,\ldots,\tilde{\rho}_{\tilde{k}}$ be the $\lambda$-bit strings recorded in $\mc{F}^{\Phi,\Delta-\Phi,\alpha}_\msf{SBC}$ that were provided by the corrupted parties during the SBC period. Let $\tilde{\msf{tag}}_1,\ldots,\tilde{\msf{tag}}_{\tilde{k}}$ be the corresponding tags. It sets \[\tilde{r}\leftarrow\bigoplus_{j=1}^{\tilde{k}}\tilde{\rho}_j.\]
    \item Let $P^*_1.\ldots,P^*_{k^*}$ be the honest parties. For $i=2,\cdots,k^*$ it chooses a random $\lambda$-bit string $\rho^*_i$. For $P^*_1$, it sets 
\[\rho^*_1\leftarrow\Big(\bigoplus_{i=2}^{k^*}\rho^*_i\Big)\oplus\tilde{r}\oplus r.\]
    Let $\msf{tag}^*_1,\ldots,\msf{tag}^*_{k^*}$ be the corresponding tags chosen at the broadcast simulation above (recall that each honest party has submitted exactly one SBC request).
    \item It sets a vector $\mbf{v}$ that contains the pairs $(\msf{tag}^*_1,\rho^*_1),\ldots,(\msf{tag}^*_{k^*},\rho^*_{k^*}),(\tilde{\msf{tag}}_1,\tilde{\rho}_1),\ldots,(\tilde{\rho}_{\tilde{k}},\tilde{\msf{tag}}_{\tilde{k}})$ ordered ``chronologically'' as dictated by the sequence of \textsc{Broadcast} requests to $\mc{F}^{\Phi,\Delta-\Phi,\alpha}_\msf{SBC}$ from all parties.
    \item It sends $(\sid,\textsc{Broadcast},\mbf{v})$ to $\mc{A}$.
\end{enumerate}

During the simulation, $\mc{S}_\msf{DURS}$ totally respects the description of $\mc{F}^{\Phi,\Delta-\Phi,\alpha}_\msf{SBC}$ and $\mc{F}_\msf{RBC}$ and the interaction of these functionalities with $\mc{A}$, except for the fact that it does not commit to the honest parties' SBC requests from the beginning, but only when the time to send the vector of broadcast messages to $\mc{A}$ has come. However, this does not affect the leakage that $\mc{A}$ gets from $\mc{F}^{\Phi,\Delta-\Phi,\alpha}_\msf{SBC}$ during the SBC period, as this is an all-zero $\lambda$-bit string in any case. Besides, by construction it holds that 
\[\Big(\bigoplus_{i=1}^{k^*}\rho^*_i\Big)\oplus\Big(\bigoplus_{j=1}^{\tilde{k}}\tilde{\rho}_j\Big)=r,\]
i.e., the XOR of all SBC $\lambda$-bit strings equals to actual URS $r$ sampled by $\mc{F}^{\Delta,\alpha}_\msf{DURS}$ and $\mc{A}$ can compute $r$ by the SBC leakage.
Thus, it suffices to show that the modified choice of honest parties randomness contributions respects the original uniform distribution of $\Pi_\msf{DURS}$. 

Since $\rho^*_2,\ldots,\rho^*_{k^*}$ are chosen uniformly at random, the above sufficient condition boils down to showing that $\rho^*_1$ also follows the uniform distribution. This directly derives from the definition of $\rho^*_1$ and the fact that $r$ is chosen uniformly at random by $\mc{F}^{\Delta,\alpha}_\msf{DURS}$.

\end{proof}

%% file: main.bbl
\newcommand{\etalchar}[1]{$^{#1}$}
\begin{thebibliography}{CDPW07}

\bibitem[ALZ21]{ALZ21}
Myrto Arapinis, Nikolaos Lamprou, and Thomas Zacharias.
\newblock Astrolabous: {A} universally composable time-lock encryption scheme.
\newblock In Mehdi Tibouchi and Huaxiong Wang, editors, {\em Advances in
  Cryptology - {ASIACRYPT} 2021 - 27th International Conference on the Theory
  and Application of Cryptology and Information Security, Singapore, December
  6-10, 2021, Proceedings, Part {II}}, volume 13091 of {\em Lecture Notes in
  Computer Science}, pages 398--426. Springer, 2021.

\bibitem[BDD{\etalchar{+}}21]{BaumDDNO21}
Carsten Baum, Bernardo David, Rafael Dowsley, Jesper~Buus Nielsen, and Sabine
  Oechsner.
\newblock {TARDIS:} {A} foundation of time-lock puzzles in {UC}.
\newblock In Anne Canteaut and Fran{\c{c}}ois{-}Xavier Standaert, editors, {\em
  Advances in Cryptology - {EUROCRYPT} 2021 - 40th Annual International
  Conference on the Theory and Applications of Cryptographic Techniques,
  Zagreb, Croatia, October 17-21, 2021, Proceedings, Part {III}}, volume 12698
  of {\em Lecture Notes in Computer Science}, pages 429--459. Springer, 2021.

\bibitem[BMTZ17]{BadertscherMTZ17}
Christian Badertscher, Ueli Maurer, Daniel Tschudi, and Vassilis Zikas.
\newblock Bitcoin as a transaction ledger: {A} composable treatment.
\newblock In Jonathan Katz and Hovav Shacham, editors, {\em Advances in
  Cryptology - {CRYPTO} 2017 - 37th Annual International Cryptology Conference,
  Santa Barbara, CA, USA, August 20-24, 2017, Proceedings, Part {I}}, volume
  10401 of {\em Lecture Notes in Computer Science}, pages 324--356. Springer,
  2017.

\bibitem[Can01]{UC}
Ran Canetti.
\newblock Universally composable security: {A} new paradigm for cryptographic
  protocols.
\newblock In {\em 42nd Annual Symposium on Foundations of Computer Science,
  {FOCS} 2001, 14-17 October 2001, Las Vegas, Nevada, {USA}}, pages 136--145.
  {IEEE} Computer Society, 2001.

\bibitem[Can04]{Canetti04}
Ran Canetti.
\newblock Universally composable signature, certification, and authentication.
\newblock In {\em 17th {IEEE} Computer Security Foundations Workshop,
  {(CSFW-17} 2004), 28-30 June 2004, Pacific Grove, CA, {USA}}, page 219.
  {IEEE} Computer Society, 2004.

\bibitem[CDPW07]{globalUC}
Ran Canetti, Yevgeniy Dodis, Rafael Pass, and Shabsi Walfish.
\newblock Universally composable security with global setup.
\newblock In Salil~P. Vadhan, editor, {\em Theory of Cryptography, 4th Theory
  of Cryptography Conference, {TCC} 2007, Amsterdam, The Netherlands, February
  21-24, 2007, Proceedings}, volume 4392 of {\em Lecture Notes in Computer
  Science}, pages 61--85. Springer, 2007.

\bibitem[CGMA85]{ChorGMA85}
Benny Chor, Shafi Goldwasser, Silvio Micali, and Baruch Awerbuch.
\newblock Verifiable secret sharing and achieving simultaneity in the presence
  of faults (extended abstract).
\newblock In {\em 26th Annual Symposium on Foundations of Computer Science,
  Portland, Oregon, USA, 21-23 October 1985}, pages 383--395. {IEEE} Computer
  Society, 1985.

\bibitem[CGZ21]{CGZ21}
Ran Cohen, Juan Garay, and Vassilis Zikas.
\newblock Completeness theorems for adaptively secure broadcast.
\newblock Cryptology ePrint Archive, Paper 2021/775, 2021.
\newblock \url{https://eprint.iacr.org/2021/775}.

\bibitem[CR87]{ChorR87}
Benny Chor and Michael~O. Rabin.
\newblock Achieving independence in logarithmic number of rounds.
\newblock In Fred~B. Schneider, editor, {\em Proceedings of the Sixth Annual
  {ACM} Symposium on Principles of Distributed Computing, Vancouver, British
  Columbia, Canada, August 10-12, 1987}, pages 260--268. {ACM}, 1987.

\bibitem[DS82]{DS82}
Danny Dolev and H.~Raymond Strong.
\newblock Polynomial algorithms for multiple processor agreement.
\newblock In Harry~R. Lewis, Barbara~B. Simons, Walter~A. Burkhard, and
  Lawrence~H. Landweber, editors, {\em Proceedings of the 14th Annual {ACM}
  Symposium on Theory of Computing, May 5-7, 1982, San Francisco, California,
  {USA}}, pages 401--407. {ACM}, 1982.

\bibitem[FKL08]{FaustKL08}
Sebastian Faust, Emilia K{\"{a}}sper, and Stefan Lucks.
\newblock Efficient simultaneous broadcast.
\newblock In Ronald Cramer, editor, {\em Public Key Cryptography - {PKC} 2008,
  11th International Workshop on Practice and Theory in Public-Key
  Cryptography, Barcelona, Spain, March 9-12, 2008. Proceedings}, volume 4939
  of {\em Lecture Notes in Computer Science}, pages 180--196. Springer, 2008.

\bibitem[Gen00]{Gennaro00}
Rosario Gennaro.
\newblock A protocol to achieve independence in constant rounds.
\newblock {\em {IEEE} Trans. Parallel Distributed Syst.}, 11(7):636--647, 2000.

\bibitem[GKKZ11]{GarayKKZ11}
Juan~A. Garay, Jonathan Katz, Ranjit Kumaresan, and Hong{-}Sheng Zhou.
\newblock Adaptively secure broadcast, revisited.
\newblock In Cyril Gavoille and Pierre Fraigniaud, editors, {\em Proceedings of
  the 30th Annual {ACM} Symposium on Principles of Distributed Computing,
  {PODC} 2011, San Jose, CA, USA, June 6-8, 2011}, pages 179--186. {ACM}, 2011.

\bibitem[GKO{\etalchar{+}}20]{GarayKOPZ20}
Juan~A. Garay, Aggelos Kiayias, Rafail~M. Ostrovsky, Giorgos Panagiotakos, and
  Vassilis Zikas.
\newblock Resource-restricted cryptography: Revisiting {MPC} bounds in the
  proof-of-work era.
\newblock In Anne Canteaut and Yuval Ishai, editors, {\em Advances in
  Cryptology - {EUROCRYPT} 2020 - 39th Annual International Conference on the
  Theory and Applications of Cryptographic Techniques, Zagreb, Croatia, May
  10-14, 2020, Proceedings, Part {II}}, volume 12106 of {\em Lecture Notes in
  Computer Science}, pages 129--158. Springer, 2020.

\bibitem[Gro04]{Groth04}
Jens Groth.
\newblock Efficient maximal privacy in boardroom voting and anonymous
  broadcast.
\newblock In Ari Juels, editor, {\em Financial Cryptography, 8th International
  Conference, {FC} 2004, Key West, FL, USA, February 9-12, 2004. Revised
  Papers}, volume 3110 of {\em Lecture Notes in Computer Science}, pages
  90--104. Springer, 2004.

\bibitem[Hev06]{Hevia06}
Alejandro Hevia.
\newblock Universally composable simultaneous broadcast.
\newblock In Roberto~De Prisco and Moti Yung, editors, {\em Security and
  Cryptography for Networks, 5th International Conference, {SCN} 2006, Maiori,
  Italy, September 6-8, 2006, Proceedings}, volume 4116 of {\em Lecture Notes
  in Computer Science}, pages 18--33. Springer, 2006.

\bibitem[HM05]{HeviaM05}
Alejandro Hevia and Daniele Micciancio.
\newblock Simultaneous broadcast revisited.
\newblock In Marcos~Kawazoe Aguilera and James Aspnes, editors, {\em
  Proceedings of the Twenty-Fourth Annual {ACM} Symposium on Principles of
  Distributed Computing, {PODC} 2005, Las Vegas, NV, USA, July 17-20, 2005},
  pages 324--333. {ACM}, 2005.

\bibitem[HZ10]{HirtZ10}
Martin Hirt and Vassilis Zikas.
\newblock Adaptively secure broadcast.
\newblock In Henri Gilbert, editor, {\em Advances in Cryptology - {EUROCRYPT}
  2010, 29th Annual International Conference on the Theory and Applications of
  Cryptographic Techniques, Monaco / French Riviera, May 30 - June 3, 2010.
  Proceedings}, volume 6110 of {\em Lecture Notes in Computer Science}, pages
  466--485. Springer, 2010.

\bibitem[KMTZ13]{KatzMTZ13}
Jonathan Katz, Ueli Maurer, Bj{\"{o}}rn Tackmann, and Vassilis Zikas.
\newblock Universally composable synchronous computation.
\newblock In Amit Sahai, editor, {\em Theory of Cryptography - 10th Theory of
  Cryptography Conference, {TCC} 2013, Tokyo, Japan, March 3-6, 2013.
  Proceedings}, volume 7785 of {\em Lecture Notes in Computer Science}, pages
  477--498. Springer, 2013.

\bibitem[KY02]{KiayiasY02}
Aggelos Kiayias and Moti Yung.
\newblock Self-tallying elections and perfect ballot secrecy.
\newblock In David Naccache and Pascal Paillier, editors, {\em Public Key
  Cryptography, 5th International Workshop on Practice and Theory in Public Key
  Cryptosystems, {PKC} 2002, Paris, France, February 12-14, 2002, Proceedings},
  volume 2274 of {\em Lecture Notes in Computer Science}, pages 141--158.
  Springer, 2002.

\bibitem[Nie02]{Nielsen02}
Jesper~Buus Nielsen.
\newblock Separating random oracle proofs from complexity theoretic proofs: The
  non-committing encryption case.
\newblock In Moti Yung, editor, {\em Advances in Cryptology - {CRYPTO} 2002,
  22nd Annual International Cryptology Conference, Santa Barbara, California,
  USA, August 18-22, 2002, Proceedings}, volume 2442 of {\em Lecture Notes in
  Computer Science}, pages 111--126. Springer, 2002.

\bibitem[PSL80]{PeaseSL80}
Marshall~C. Pease, Robert~E. Shostak, and Leslie Lamport.
\newblock Reaching agreement in the presence of faults.
\newblock {\em J. {ACM}}, 27(2):228--234, 1980.

\bibitem[SP15]{SzepieniecP15}
Alan Szepieniec and Bart Preneel.
\newblock New techniques for electronic voting.
\newblock {\em {IACR} Cryptol. ePrint Arch.}, page 809, 2015.

\end{thebibliography}
